\documentclass[12pt,draftclsnofoot, onecolumn]{IEEEtran}
\usepackage{soul}
\usepackage{graphicx}
\usepackage{mathtools}
\usepackage{amsmath}
\usepackage{amsfonts}
\usepackage{array}
\newcommand{\PreserveBackslash}[1]{\let\temp=\\#1\let\\=\temp}
\newcolumntype{C}[1]{>{\PreserveBackslash\centering}p{#1}}
\newcolumntype{R}[1]{>{\PreserveBackslash\raggedleft}p{#1}}
\newcolumntype{L}[1]{>{\PreserveBackslash\raggedright}p{#1}}
\usepackage[usenames]{color}
\usepackage{colortbl,booktabs}
\usepackage{tabularx}
\usepackage{multicol}
\usepackage{booktabs}
\usepackage[Symbol]{upgreek}
\usepackage{subfigure}
\usepackage{bm}
\usepackage{url}
\usepackage[usenames,dvipsnames,svgnames,table]{xcolor}

\usepackage{cite}
\usepackage{balance}
\usepackage{times}
\newtheorem{lemma}{Lemma}
\newtheorem{proof}{Proof}
\newtheorem{theorem}{Theorem}

\usepackage{stfloats}
\usepackage{balance}

\usepackage[nonumberlist,acronym,shortcuts]{glossaries}
\makeglossaries
\makeindex
\begin{document}

\title{\Large Regularized Zero-Forcing Precoding Aided Adaptive Coding and Modulation
 for Large-Scale Antenna Array Based Air-to-Air Communications}

\author{Jiankang~Zhang,~\IEEEmembership{Senior Member,~IEEE},
 Sheng~Chen,~\IEEEmembership{Fellow,~IEEE},
 Robert~G.~Maunder,~\IEEEmembership{Senior~Member,~IEEE},
 Rong~Zhang,~\IEEEmembership{Senior~Member,~IEEE},
 Lajos~Hanzo,~\IEEEmembership{Fellow,~IEEE}
\thanks{The authors are with School of Electronics and Computer Science, University
 of Southampton, U.K. (E-mails: \{jz09v, sqc, rm, rz, lh\}@ecs.soton.ac.uk. S. Chen
 is also with King Abdulaziz University, Jeddah, Saudi Arabia.} %
\thanks{The financial support of the European Research Council's Advanced Fellow 
 Grant and of the Royal Society Wolfson Research Merit Award as well as of the EPSRC project
 EP/N004558/1 are gratefully acknowledged. The research
data for this paper is available at https://doi.org/10.5258/SOTON/D0592.} %
\vspace*{-15mm}
}

\maketitle

\IEEEpeerreviewmaketitle

\begin{abstract}
 We propose a regularized zero-forcing transmit precoding (RZF-TPC) aided and
 distance-based adaptive coding and modulation (ACM) scheme to support aeronautical
 communication applications, by exploiting the high spectral efficiency of
 large-scale antenna arrays and link adaption. Our RZF-TPC aided and distance-based
 ACM scheme switches its mode according to the distance between the communicating
 aircraft. We derive the closed-form asymptotic signal-to-interference-plus-noise
 ratio (SINR) expression of the RZF-TPC for the aeronautical channel, which is
 Rician, relying on a non-centered channel matrix that is dominated by the
 deterministic line-of-sight component. The effects of both realistic channel
 estimation errors and of the co-channel interference are considered in the
 derivation of this approximate closed-form SINR formula. Furthermore, we derive
 the analytical expression of the optimal regularization parameter that minimizes
 the mean square detection error. The achievable throughput expression based on our 
 asymptotic approximate SINR formula is then utilized as the design metric for the
 proposed RZF-TPC aided and distance-based ACM scheme. Monte-Carlo simulation
 results are presented for validating our theoretical analysis as well as for
 investigating the impact of the key system parameters. The simulation results
 closely match  the theoretical results. In the specific example that two
 communicating aircraft fly at a typical cruising speed of 920\,km/h, heading in
 opposite direction over the distance up to 740\,km taking a period of about 24
 minutes, the RZF-TPC aided and distance-based ACM is capable of transmitting a
 total of 77 Gigabyte of data with the aid of 64 transmit antennas and 4 receive
 antennas, which is significantly higher than that of our previous
 eigen-beamforming transmit precoding aided and distance-based ACM benchmark.
\end{abstract}


\section{Introduction}\label{S1}

 The vision of the `smart sky' \cite{zhang2017survey} in support of air traffic control
 and the `Internet above the clouds' \cite{jahn2003evolution} for in-flight entertainment
 has motivated researchers to develop new solutions for aeronautical communications. The
 aeronautical {\it{ad hoc}} network (AANET) \cite{vey2014aeronautical} exchanges
 information using multi-hop air-to-air radio communication links, which is capable of
 substantially extending the coverage range over the oceanic and remote airspace, without
 any additional infrastructure and without relying on satellites. However, the
 existing air-to-air communication solutions can only provide limited data rates.
 Explicitly, the planed L-band digital aeronautical communications system (L-DACS)
 \cite{schnell2014ldacs,jain2011analysis} only provides upto 1.37\,Mbps air-to-ground
 communication rate, and the aeronautical mobile airport communication system
 \cite{bartoli2013aeromacs} only offers 9.2 Mbps air-to-ground communication rate in the
 vicinity of the airport. Finally, the L-DACS air-to-air mode \cite{graupl2011ldacs1} is
 only capable of providing 273\,kbps net user rate for direct air-to-air communication,
 which cannot meet the high-rate demands of the emerging aeronautical applications.

 The existing aeronautical communication systems mainly operate in the very high
 frequency band spanning from 118\,MHz to 137\,MHz \cite{Haind2007Anindependent},
 and there are no substantial idle frequency slots for developing broadband commercial
 aeronautical communications. Moreover, the ultra high frequency band has almost
 been fully occupied by television broadcasting, cell phones and satellite
 communications \cite{zhang2017survey,stacey2008aeronautical}. However, there are many
 unlicensed-frequencies in the super high frequency (SHF) band spanning from 3\,GHz to
 30\,GHz, which may be explored for the sake of developing broadband commercial
 aeronautical communications. Explicitly, the wavelength spans from 1\,cm to 10\,cm for
 the SHF band, which results in 0.5\,cm\,$\sim$\,5\,cm antenna spacing by utilizing the
 half-wavelength criterion for designing the antenna array. This antenna spacing is
 capable of accommodating a large-scale antenna array on commercial aircraft, which 
 offers dramatic  throughput and energy efficiency benefits \cite{larsson2014massive}.
 To provide a high throughput and a high spectral efficiency (SE) for commercial
 air-to-air applications, we propose a large-scale antenna array aided adaptive coding
 and modulation (ACM) based solution in the SHF band. 

 As an efficient link adaptation technique, ACM \cite{goldsmith1998adaptive,hanzo2002adaptive}
 adaptively matches the modulation and coding modes to the conditions of the propagation
 link, which is capable of enhancing the link reliability and maximizing the throughput.
 The traditional ACM relies on the instantaneous signal-to-noise ratio (SNR) or
 signal-to-interference-plus-noise ratio (SINR) to switch the ACM modes, which requires
 the acquisition of the instantaneous channel state information (CSI). Naturally, channel
 estimation errors are unavoidable in practice, especially at aircraft velocities 
 \cite{zhou2004adaptive}. Furthermore, the CSI-feedback based ACM solution may
 potentially introduce feedback errors and delays \cite{zhou2004accurate}. Intensive
 investigations have been invested in robust ACM, relying on partial CSI
 \cite{zhou2004adaptive} and imperfect CSI \cite{taki2014adaptive}, or  exploiting
 non-coherent detection for dispensing with channel estimation all together
 \cite{hanzo2004quadrature}. However, all these ACM solutions are designed for terrestrial
 wireless communications and they have to frequently calculate the SINR and to promptly
 change the ACM modes, which imposes heavy mode-signaling   overhead. Therefore, for
 air-to-air communications, these ACM designs may become impractical.

 Unlike terrestrial channels, which typically exhibit Rayleigh characteristics,
 aeronautical communication channels exhibit strong line-of-sight (LOS) propagation
 characteristics \cite{haas2002aeronautical,meng2011measurements}, and at cruising
 altitudes, the LOS component dominates the reflected components. Furthermore, the
 passenger planes typically fly across large-scale geographical distances, and the received
 signal strength is primarily determined by the pathloss, which is a function of
 communication distance. In \cite{zhang2017adaptive}, we proposed an eigen-beamforming
 transmit precoding (EB-TPC) aided and distance-based ACM solution for air-to-air
 aeronautical communication by exploiting the aeronautical channel characteristics.
 EB-TPC has the advantage of low-complexity operation by simply conjugating the channel
 matrix, and it also enables us to derive the closed-form expression of the attainable
 throughput, which facilitates the design of the distance-based ACM
 \cite{zhang2017adaptive}. However, its achievable  throughput is far from optimal,
 since EB-TPC does not actively suppress the inter-antenna interference. Zero-forcing
 transmit precoding (ZF-TPC) \cite{wiesel2008zero} by contrast is capable of mitigating
 the inter-antenna interference, but it is challenging to provide a closed-form
 expression for the achievable throughput, particularly for large-scale antenna array
 based systems.  Tataria {\it et al.} \cite{tataria2017zero} investigated the distribution
 of the instantaneous per-terminal SNR for the ZF-TPC aided multi-user system and
 approximated it as a gamma distribution. Additionally, ZF-TPC also surfers from rate
 degradation in ill-conditioned channels. By introducing regularization, the regularized
 ZF-TPC (RZF-TPC) \cite{peel2005avector} is capable of mitigating the ill-conditioning
 problem by beneficially  balancing the interference cancellation and the noise
 enhancement \cite{zhang2013large}. Furthermore, owing to the regularization, it becomes
 possible to analyze the achievable throughput for the Rayleigh fading channel. Hoydis
 {\it et al.} \cite{hoydis2013massive} used the RZF-TPC as the benchmark to study how
 many extra antennas are needed for the EB-TPC in the context of Rayleigh fading channels. 

 However, the Rician fading channel experienced in aeronautical communications, which
 has a non-centered channel matrix due to the presence of the deterministic LOS
 component, is different from the centered Rayleigh fading channel. This imposes a
 challenge on deriving a closed-form formula of the achievable throughput, which is a
 fundamental metric of designing ACM solutions. Few researches have tackled this
 challenge. Nonetheless, recently three conference papers
 \cite{tataria2016performance,falconet2016asymptotic,Sanguinetti2017asymptotic} have 
 investigated the asymptotic sum-rate of the RZF-TPC in Rician channels. Explicitly,
  Tataria {\it et al.} \cite{tataria2016performance} investigated the ergodic sum-rate of
 the RZF-TPC aided single-cell system under the idealistic condition  of uncorrelated
 Rician channel and the idealistic assumption of perfect channel knowledge. Falconet
 {\it et al.} \cite{falconet2016asymptotic} provided an asymptotic sum-rate expression for RZF-TPC in a single-cell scenario by assuming identical fading-correlation for
 all the users. Sanguinetti {\it et al.} \cite{Sanguinetti2017asymptotic} extended this 
 work from the single-cell to the coordinated multi-cell scenario under the same
 assumption. But crucially, the authors of \cite{Sanguinetti2017asymptotic} did not
 consider the pilot contamination imposed by adjacent cells during the uplink channel
 estimation \cite{zhang2014pilot,guo2016optimal}. Moreover, the study
 \cite{Sanguinetti2017asymptotic} assumed Rician fading only within the serving cell,
 while the interfering signals arriving from adjacent cells were still assumed to suffer
 from Rayleigh fading. This assumption has limited validity in aeronautical communications.
 Most critically, the asymptotic sum-rates provided in \cite{falconet2016asymptotic} and
 \cite{Sanguinetti2017asymptotic} were based on the assumption that both the number
 of antennas and the number of served users tend to infinity. The essence of
 the `massive' antenna array systems is that of serving a small number of users on the same
 resource block using linear signal processing by employing a large number
 of antenna elements. Assuming that the number of users on a resource block tends to
 infinity has no physical foundation at all.

 Against this background, this paper designs an RZF-TPC scheme for  large-scale antenna
 array assisted and distance-based ACM aided aeronautical communications, which offers
 an appealing solution for supporting the emerging Internet above the clouds. Our main
 contributions are:
\begin{enumerate}
\item We derive the closed-form expression of the achievable throughput for the RZF-TPC
 in the challenging new context of aeronautical communications. Our previous contribution work relying on EB-TPC  \cite{zhang2017adaptive} invoked relatively simple analysis, since it did not involve the non-centered channel matrix inverse. By contrast, the derivation of the closed-form throughput of our new RZF-TPC has to tackle the associated non-centered matrix inverse problem. Moreover, in contrast to the EB-TPC, the  regularization parameter of the RZF-TPC has to be optimized for maximizing the throughput. In this paper,  
 we derive the closed-form asymptotic approximation of the SINR for the RZF-TPC in the
 presence of both realistic channel estimation errors and co-channel interference
 imposed by the aircraft operating in the same frequency band.  We also provide the associated detailed proof.  Moreover, we explicitly
 derive the optimal analytical regularization parameter that minimizes the mean square
 detection error. Given this asymptotic approximation of the SINR, the fundamental metric
 of the achievable throughput as the function of the communication distance is provided
 for designing the distance-based ACM.

\item We develop the new RZF-TPC aided and distance-based ACM design for the application
 to the large antenna array assisted aeronautical communication in the presence of
 imperfect CSI and co-channel interference,  first considered in \cite{zhang2017adaptive}.
  Like our previous EB-TPC aided and distance-based ACM scheme \cite{zhang2017adaptive},
 the RZF-TPC aided and distance-based ACM scheme switches its ACM mode based on the
 distance between the communicating aircraft pair. However, the RZF-TPC is much more
 powerful, and the proposed design offers significantly higher SE over the previous
 EB-TPC aided and distance-based ACM design.  Specifically, the new design achieves up
 to 3.0 bps/Hz and 3.5 bps/Hz SE gains with the aid of 32 transmit antennas/4 receive
 antennas and 64 transmit antennas/4 receive antennas, respectively, over our previous
 design.
\end{enumerate}

\section{System Model}\label{S2}

 We consider an air-to-air communication scenario at cruising altitude. Our
 proposed time division duplex (TDD) based aeronautical communication system is
 illustrated Fig.~\ref{FIG1}. In the communication zone considered, aircraft $a^*$ 
 transmits its data to aircraft $b^*$, while aircraft $a$, $a = 1,2,\cdots,A$ are
 the interfering  aircraft using the same frequency as aircraft $a^*$ and $b^*$. The
 aeronautical communication system operates in the SHF band and we assume  that the
 carrier frequency is 5\,GHz, which results in a wave-length of 6\,cm. Thus, it is 
 practical to accommodate a large-scale high-gain antenna array on the aircraft for
 achieving high SE. We assume furthermore that all the aircraft are equipped with the
 same large-scale antenna array. Specifically, each aircraft has $N_{\rm total}$
 antennas, which transmit and receive signals on the same frequency. Explicitly, each
 aircraft utilizes $N_t$ ($< N_{\rm total}$) antennas, denoted as data-transmitting
 antennas (DTAs), for transmitting data and utilizes $N_r$ antennas, denoted as
 data-receiving antennas (DRAs), for receiving data. In line with  the maximum
 attainable spatial degrees of freedom, generally, we have $N_r < N_t$.  Furthermore, the
 system adopts orthogonal frequency-division multiplexing (OFDM) for improving the SE and
 the TDD protocol for reducing the latency imposed by channel information feedback.
 Each aircraft has a distance measuring equipment (DME), e.g., radar, which is
 capable of measuring the distance to nearby aircraft. Alternatively, the GPS system 
 may be utilized to provide the distance information required.

\begin{figure*}[tp!]
\vspace*{-1mm}
\begin{center}
 \includegraphics[width=0.65\textwidth]{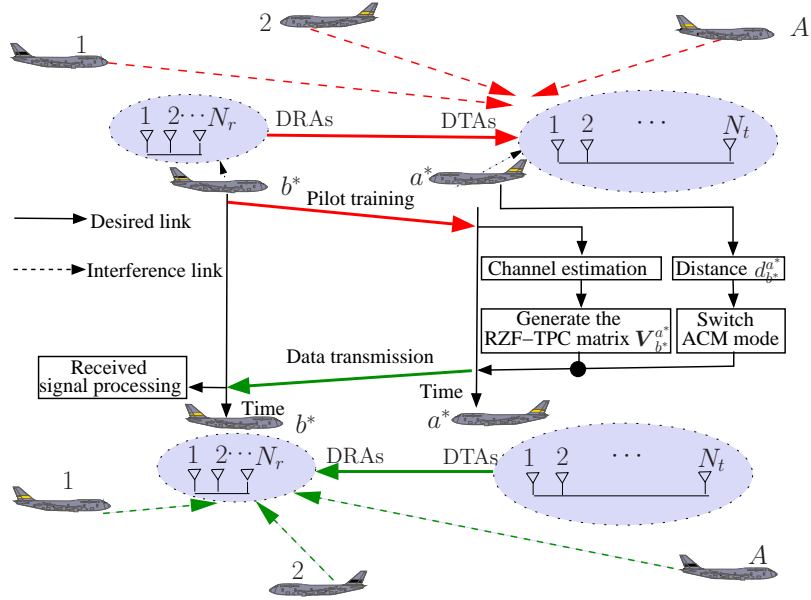}
\end{center}
\vspace*{-10mm}
\caption{The proposed aeronautical communication system employing the RZF-TPC aided and
 distance-based ACM scheme, where aircraft $a^*$ is transmitting data to aircraft $b^*$
 in the presence of co-channel interference.}
\label{FIG1}
\vspace*{-6mm}
\end{figure*}

\subsection{Channel State Information Acquisition}\label{S2.1}

 In order to transmit data from $a^*$ to $b^*$, aircraft $a^*$ needs the CSI linking
 $a^*$ to aircraft $b^*$. Aircraft $a^*$ estimates the reverse channel based on the
 pilots sent by $b^*$, and then exploits the channel's reciprocity of TDD protocol to
 acquire the required CSI. Explicitly, this pilot training phase is shown at the top
 of Fig.~\ref{FIG1}, where $a^*$ estimates the channel between the $N_r$ DRAs of
 $b^*$ and its $N_t$ DTAs based on the pilots sent by $b^*$ in the presence of the
 interference imposed by the aircraft $a$, $a=1,2,\cdots,A$. We consider the
 worst-case scenario, where the interfering aircraft $a$ also transmits the same pilot
 symbols as $b^*$, which results in the most serious co-channel interference. Since the
 length of the cyclic prefix (CP) $N_{\rm cp}$ is longer than the channel length $P$,
 inter-symbol interference is completely eliminated, and the receiver can process the
 signals on a subcarrier-by-subcarrier basis. Thus, the frequency-domain (FD) signal
 vector of $a^{*}$, $\widetilde{\bm{Y}}_{a^*}=\big[\widetilde{Y}^{a^*}_1 ~
 \widetilde{Y}^{a^*}_2 \cdots \widetilde{Y}^{a^*}_{N_t}\big]^{\rm T}\in\mathbb{C}^{N_t}$,
 received during the pilot training can be written as 
\begin{align}\label{eq1}
 & \widetilde{\bm{Y}}_{a^*}= \sqrt{P_{r,a^*}^{b^*}} \bm{H}_{a^*}^{b^*}
  \widetilde{\bm{X}}^{b^*}  + \sum\limits_{a=1}^A \sqrt{P_{r,a^*}^a} \bm{H}_{a^*}^a
  \widetilde{\bm{X}}^a + \widetilde{\bm{W}}_{a^*} ,
\end{align}
 where $\widetilde{\bm{X}}^{b^*}=\big[\widetilde{X}^{b^*}_1 ~ \widetilde{X}^{b^*}_2
 \cdots \widetilde{X}^{b^*}_{N_r}\big]^{\rm T}\in \mathbb{C}^{N_r}$ is the pilot
 symbol vector transmitted by $b^*$, which obeys the complex Gaussian distribution
 with the mean vector of the $N_r$-dimensional zero vector $\bm{0}_{N_r}$ and the
 covariance matrix of the $N_r\times N_r$ identity matrix $\bm{I}_{N_r}$, denoted by
 $\widetilde{\bm{X}}^{b^*}\sim\mathcal{CN}\left(\bm{0}_{N_r},\bm{I}_{N_r}\right)$, and
 $\bm{H}_{a^*}^{a'}\in\mathbb{C}^{N_t\times N_r}$ denotes the FD channel transfer
 function coefficient matrix linking the $N_r$ DRAs of $a'$ to the $N_t$ DTAs of $a^*$,
 for $a'=b^*,a$, while $\widetilde{\bm{W}}_{a^*}\sim \mathcal{CN}\left(\bm{0}_{N_t},
 \sigma_w^2\bm{I}_{N_t}\right)$ is the FD additive white Gaussian noise (AWGN) vector,
 and $P_{r,a^*}^{b^*}$ and $P_{r,a^*}^a$ represent the received powers at a single DTA
 of $a^*$ for the signals transmitted from $b^*$ and $a$, respectively. Moreover, since
 the worst-case scenario is considered, aircraft $a$ uses the same pilot symbol as $b^*$,
 and we have $\widetilde{\bm{X}}^a=\widetilde{\bm{X}}^{b^*}$ for $1\le a\le A$. 
 
 Typically, the aeronautical channel consists of a strong LOS path and a cluster of
 reflected/delayed paths \cite{haas2002aeronautical,bello1973aeronautical,walter2011the}.
 Hence, the channel is Rician, and $\bm{H}_{a^*}^{b^*}\in \mathbb{C}^{N_t\times N_r}$ is
 given by 
\begin{align}\label{eq2}
 \bm{H}_{a^*}^{b^*} =& \nu \bm{H}_{\text{d},a^*}^{b^*} +
  \varsigma \bm{H}_{\text{r},a^*}^{b^*} ,
\end{align}
 where $\bm{H}_{{\rm d},a^*}^{b^*}\in \mathbb{C}^{N_t\times N_r}$ and $\bm{H}_{{\rm r},a^*}^{b^*}
 \in \mathbb{C}^{N_t\times N_r}$ are the deterministic and scattered channel components,
 respectively, while $\nu =\sqrt{\frac{K_{\rm Rice}}{K_{\rm Rice}+1}}$ and $\varsigma =
 \sqrt{1-\nu}$, in which $K_{\rm Rice}$ is the Rician $K$-factor of the channel.
  When aircraft are at cruising altitude, the deterministic LOS component dominates, and the
 scattered component is very weak which may come from the reflections from other distant
 aircraft or tall mountains. Note that when an aircraft is at cruising altitude, there
 is no local scatters at all, because a minimum safe distance is enforced among aircraft,
 and there exists no shadowing effect either. For an aircraft near airport space for
 landing/takeoff, the scattering component is much stronger than at cruising, but the LOS
 component still dominates. The scattering component in this case includes reflections
 from ground, and shadowing effect has to be considered. The scattered component
 $\bm{H}_{\text{r},a^*}^{b^*}$ can be expressed as \cite{kim2010spatial}
\begin{align}\label{eq3}
 \bm{H}_{\text{r},a^*}^{b^*} =& \bm{R}_{a^*}^{\frac{1}{2}} \bm{G}_{a^*}^{b^*}
 \big(\bm{R}^{b^*}\big)^{\frac{1}{2}} ,
\end{align}
 where $\bm{R}^{b^*}\in \mathbb{C}^{N_r\times N_r}$ and $\bm{R}_{a^*}\in
 \mathbb{C}^{N_t\times N_t}$ are the spatial correlation matrices for the $N_r$
 antennas of $b^*$ and the $N_t$ antennas of $a^*$, respectively, while the elements
 of $\bm{G}_{a^*}^{b^*}\in \mathbb{C}^{N_t\times N_r}$ follow the independently
 identically distributed distribution $\mathcal{CN}(0,1)$. Thus, $\mathcal{E}\left\{
 \bm{vec}\left(\bm{H}_{{\rm r},a^*}^{b^*}\right)\right\}=\bm{0}_{N_tN_r}$, where 
 $\mathcal{E}\{\cdot\}$ is the expectation operator and $\bm{vec}(\bm{H})$ denotes
 the column stacking operation applied to $\bm{H}$, while the covariance matrix
 $\bm{R}_{{\rm r},a^*}^{b^*}=\mathcal{E}\left\{\bm{vec}\left(\bm{H}_{{\rm r},a^*}^{b^*}\right)
 \bm{vec}^{\rm H}\left(\bm{H}_{{\rm r},a^*}^{b^*}\right)\right\}\in \mathbb{C}^{N_tN_r
 \times N_tN_r}$ is given by $\bm{R}_{{\rm r},a^*}^{b^*}=\bm{R}^{b^*}\otimes \bm{R}_{a^*}$,
 in which $\otimes$ is the Kronecker product. Since all the aircraft are assumed to be equipped
 with the same antenna array, we will assume that all the $\bm{R}_{a_t}$, $\forall a_t
 \in \mathcal{A}=\{1,2,\cdots ,A,a^*,b^*\}$, are equal\footnote{ The local scattering
 in the aeronautical channel is not as rich as in the terrestrial channel
 \cite{tataria2017impact}, and the difference in the local scatterings amongst
 different aircraft may be omitted. Furthermore, at the cruising altitude, there exists
 no local scattering at all. However, even though it is reasonable to assume that all
 jumbo jets are equipped with identical antenna arrays, the geometric shapes of different
 types of jumbo jets are slightly different, and thus $\bm{R}_{a_t}=\bar{\bm{R}}_t$,
 $\forall a_t\in\mathcal{A}$ only holds approximately.}, i.e., we have $\bm{R}_{a_t}=\bar{\bm{R}}_t$,
 $\forall a_t\in\mathcal{A}$, and all the $\bm{R}^{a_r}$ are equal, namely, $\bm{R}^{a_r}
 =\bar{\bm{R}}^r$, $\forall a_r\in\mathcal{A}$. Hence, all the covariance matrices are
 equal, and they can be expressed as
\begin{align}\label{eq4}
 \bm{R}_{\text{r},a_t}^{a_r} =& \bar{\bm{R}}_{\text{r},t}^{r} = \bar{\bm{R}}^r \otimes
  \bar{\bm{R}}_t, \, \forall a_t,a_r\in \mathcal{A} \text{ and } a_t\neq a_r .
\end{align}
 Note that in practice, $N_r \ll N_t$ and, therefore, the DRAs can always be spaced
 sufficiently apart so that they become uncorrelated. Consequently, we have
 $\bar{\bm{R}}^r=\bm{I}_{N_r}$.

 According to \cite{parsons2000themobile}, the received power $P_{r,a^*}^{b^*}$ at
 a single DTA antenna of aircraft $a^*$ is related to the transmitted signal power
 $P_t^{b^*}$ at a single DRA antenna of $b^*$ by
\begin{equation}\label{eq5}
 P_{r,a^*}^{b^*} = P_t^{b^*} 10^{-0.1 L_{{\rm path loss},a^*}^{b^*}} .
\end{equation}
 Since we mainly consider air-to-air transmissions, there exists no shadowing, and 
 the pathloss model can be expressed as \cite{parsons2000themobile}
\begin{equation}\label{eq6}
 L_{\text{path loss},a^*}^{b^*} \,[\text{dB}] = -154.06 + 20\log_{10}\left(f\right) + 20\log_{10}\left(d\right) ,
\end{equation}
 where $f$ [Hz] is the carrier frequency and $d$ [m] is the distance between the
 communicating aircraft pair. For the received interference signal power $P_{r,a^*}^a$,
 we have a similar pathloss model.  For air-to-ground communication near airport space,
 it may need to consider shadowing effect, and the shadow fading standard deviation in
 dB should be added to the pathloss model \cite{gligorevic2013measurements}.
 
 The minimum mean square error (MMSE) estimate $\widehat{\bm{H}}_{a^*}^{b^*}$ of
 $\bm{H}_{a^*}^{b^*}$ is given by \cite{kay2003fundamentals}
\begin{align}\label{eq7}
  \bm{vec}\left(\widehat{\bm{H}}_{a^*}^{b^*}\right) =& \bm{vec}\left(\nu\bm{H}_{\text{d},a^*}^{b^*}\right)
  + \varsigma^2\bar{\bm{R}}_{\text{r},t}^r \left(\frac{\sigma_w^2}{P_{r,a^*}^{b^*}} \bm{I}_{N_r N_t}
  + \varsigma^2\bar{\bm{R}}_{\text{r},t}^r
  + \sum\limits_{a=1}^A \frac{P_{r,a^*}^a}{P_{r,a^*}^{b^*}} \varsigma^2 \bar{\bm{R}}_{\text{r},t}^r \right)^{-1} 
  \nonumber \\ & \hspace*{-10mm}\times \bigg( \bm{vec}\left(\varsigma \bm{H}_{\text{r},a^*}^{b^*}\right)
  + \sum\limits_{a=1}^A\sqrt{\frac{P_{r,a^*}^a}{P_{r,a^*}^{b^*}}}
  \bm{vec}\left(\varsigma\bm{H}_{\text{r},a^*}^a\right) 
  + \frac{1}{\sqrt{P_{r,a^*}^{b^*}}} \bm{vec}\left(\widetilde{\bar{\bm{W}}}_{a^*}
  \Big(\widetilde{\bar{\bm{X}}}^{b^*}\Big)^{\rm H} \right) \bigg) \! ,
\end{align}
 where $\widetilde{\bar{\bm{X}}}^{b^*}\in \mathbb{C}^{N_r\times N_r}$ consists of the $N_r$
 consecutive pilot symbols with $\widetilde{\bar{\bm{X}}}^{b^*}
 \big(\widetilde{\bar{\bm{X}}}^{b^*}\big)^{\rm H}=\bm{I}_{N_r}$, and
 $\widetilde{\bar{\bm{W}}}_{a^*}\in \mathbb{C}^{N_r\times N_r}$ is the corresponding
 AWGN matrix over the $N_r$ consecutive OFDM symbols. Explicitly, the distribution of
 the MMSE estimator (\ref{eq7}) is \cite{kay2003fundamentals}
\begin{align}\label{eq8}
 \bm{vec}\left(\widehat{\bm{H}}_{a^*}^{b^*}\right) \sim \mathcal{CN}\left(
 \bm{vec}\left(\nu\bm{H}_{\text{d},a^{*}}^{b^*}\right),\bm{\Phi}_{a^{*}}^{b^{*}}\right) ,
\end{align}
 whose covariance matrix $\bm{\Phi}_{a^{*}}^{b^{*}}\in\mathbb{C}^{N_{t}N_{r}\times N_{t}N_{r}}$
 is given by 
\begin{align}\label{eq9}
 \bm{\Phi}_{a^*}^{b^*} =& \varsigma^2 \bar{\bm{R}}_{\text{r},t}^r \left(
  \frac{\sigma_w^2}{P_{r,a^{*}}^{b^{*}}} \bm{I}_{N_r N_t} \! +\! \varsigma^2 \bar{\bm{R}}_{\text{r},t}^r
  \! + \! \sum\limits_{a=1}^A \frac{P_{r,a^{*}}^a}{P_{r,a^{*}}^{b^{*}}} \varsigma^2
  \bar{\bm{R}}_{\text{r},t}^r\right)^{-1}   \varsigma^2 \bar{\bm{R}}_{\text{r},t}^r .
\end{align}
 By defining $\bm{\Phi}_{a^*}=\mathcal{E}\left\{\widehat{\bm{H}}_{{\rm r},a^*}^{b^*}
 \left(\widehat{\bm{H}}_{{\rm r},a^*}^{b^*}\right)^{\rm H}\right\}\in \mathbb{C}^{N_t\times N_t}$
 and $\bm{\Phi}^{b^*}=\mathcal{E}\left\{\left(\widehat{\bm{H}}_{{\rm r},a^*}^{b^*}\right)^{\rm H}
 \widehat{\bm{H}}_{{\rm r},a^*}^{b^*}\right\}\in \mathbb{C}^{N_r\times N_r}$, where
 $\widehat{\bm{H}}_{{\rm r},a^*}^{b^*}$ denotes the estimate of $\bm{H}_{{\rm r},b^*}^{a^*}$,
 $\bm{\Phi}_{a^*}^{b^*}$ can be expressed as
\begin{align}\label{eq10}
 \bm{\Phi}_{a^*}^{b^*} =& \bm{\Phi}_{a^*} \otimes \bm{\Phi}^{b^*} .
\end{align}
 According to Lemma~1 of \cite{fernandes2013inter}, $\bm{\Phi}^{b^*}\to \bm{I}_{N_r}$ as
 $N_t\to \infty$. Since $N_t$ is large, we have $\bm{\Phi}^{b^*}\approx \bm{I}_{N_r}$. Hence,
 given $\bm{\Phi}_{a^*}^{b^*}$, $\bm{\Phi}_{a^*}$ is uniquely determined.  It is well known
 that the computational complexity of this optimal MMSE channel estimator is on the order
 of $\textsf{O}\big(N_r^3N_t^3\big)$.

\subsection{Data Transmission}\label{S2.2}

 During the data transmission, $a^*$ transmits the data vector $\bm{X}^{a^*}=
 \big[ X_1^{a^*} ~ X_2^{a^*}\cdots X_{N_r}^{a^*}\big]^{\rm T}\in\mathbb{C}^{N_r}$
 using its $N_t$ DTAs to the $N_r$ DRAs of $b^*$, in the presence of the co-channel
 interference imposed by other aircraft, as shown at the bottom of Fig.~\ref{FIG1}.
 Owing to the TDD channel reciprocity, the channel $\bm{H}_{b^*}^{a^*}\in
 \mathbb{C}^{N_r\times N_t}$ encountered by transmitting $\bm{X}^{a^*}$ is
 $\bm{H}_{b^*}^{a^*}=\left(\bm{H}_{a^*}^{b^*}\right)^{\rm H}$ and its estimate is
 given by $\widehat{\bm{H}}_{b^*}^{a^*}=\left(\widehat{\bm{H}}_{a^*}^{b^*}\right)^{\rm H}$,
 which is used for designing the transmit precoding (TPC) for mitigating the
 inter-antenna interference (IAI). We adopt the powerful RZF-TPC whose TPC matrix
 $\bm{V}_{b^*}^{a^*}\in\mathbb{C}^{N_t\times N_r}$ is given by 
\begin{align}\label{eq11}
 \bm{V}_{b^*}^{a^*} =& \bm{\Upsilon}_{b^*}^{a^*}(\widehat{\bm{H}}_{b^*}^{a^*})^{\rm H} ,
\end{align}
 with
\begin{align}\label{eq12}
 \bm{\Upsilon}_{b^*}^{a^*} =& \left(\frac{1}{N_{t}}(\widehat{\bm{H}}_{b^*}^{a^*})^{\rm H}
  \widehat{\bm{H}}_{b^*}^{a^*} + \xi_{b^*}^{a^*}\bm{I}_{N_{t}}\right)^{-1} ,
\end{align}
 where $\xi_{b^*}^{a^*} > 0$ is the regularization parameter.  It can be seen that the 
 complexity of calculating the TPC matrix for the RZF-TPC scheme is on the order of
 $\textsf{O}\big(N_t^3\big)$. Given $\bm{V}_{b^*}^{a^*}$, the received signal vector
 $\bm{Y}_{b^*}\in \mathbb{C}^{N_r}$ of aircraft $b^*$ can be written as
\begin{align}\label{eq13}
 \bm{Y}_{b^*} =& \sqrt{P_{r,b^*}^{a^*}}\bm{H}_{b^*}^{a^*}\bm{V}_{b^*}^{a^*}
  \bm{X}^{a^*} \! +\! \sum\limits_{a=1}^A \sqrt{P_{r,b^*}^a}\bm{H}_{b^*}^a\bm{V}_{b^a}^a \bm{X}^a
  \! +\! \bm{W}_{b^*} ,
\end{align}
 where aircraft $a$ uses the RZF-TPC matrix $\bm{V}_{b^a}^a\in \mathbb{C}^{N_t\times N_r}$ to
 transmit the data vector $\bm{X}^a=\big[X_1^a ~ X_2^a \cdots X_{N_r}^a\big]^{\rm T}$ to its
 desired receiving aircraft $b^a$ for $1\le a\le A$, $b^a\ne b^*$ and $b^a\neq a$, and hence
 $\sqrt{P_{r,b^*}^a}\bm{H}_{b^*}^a\bm{V}_{b^a}^a \bm{X}^a$ is the interference imposed by $a$,
 while the AWGN vector $\bm{W}_{b^*}=\big[W_1^{b^*} ~ W_2^{b^*} \cdots W_{N_r}^{b^*}\big]^{\rm T}$
 has the distribution $\mathcal{CN}\big(\bm{0}_{N_r},\sigma_w^2\bm{I}_{N_r}\big)$.
 By using $[\bm{A}]_{[n:~]}$ and $[\bm{A}]_{[~:m]}$ to denote the $n$-th row and $m$-th column
 of $\bm{A}$, respectively, the signal received by the $n_r^{*}$-th antenna of aircraft $b^*$
 can be expressed as
\begin{align}\label{eq14}
 Y^{b^*}_{n_r^{*}} =& \sqrt{P_{r,b^*}^{a^*}} \left[\bm{H}_{b^*}^{a^*}\right]_{[n_r^{*}:~]}
  \left[\bm{V}_{b^*}^{a^*}\right]_{[~:n_r^{*}]} X_{n_r^{*}}^{a^*} 
  + \sum\limits_{n_r \neq n_r^{*}}\sqrt{P_{r,b^*}^{a^*}} \left[\bm{H}_{b^*}^{a^*}\right]_{[n_r^{*}:~]}
  \left[\bm{V}_{b^*}^{a^*}\right]_{[~:n_r]} X_{n_r}^{a^*} \nonumber \\ &
  + \sum\limits_{a=1}^A\sum\limits_{n_r = 1}^{N_{r}}\sqrt{P_{r,b^*}^a}\left[\bm{H}_{b^*}^a\right]_{[n_r^{*}:~]}
  \left[\bm{V}_{b^a}^a\right]_{[~:n_r]} X_{n_r}^{a} + W^{b^*}_{n_r^{*}}.  
\end{align}
 where the first term in the right-hand side of (\ref{eq14}) is the desired signal, the second
 term represents the IAI imposed by the $n_r$ antennas of aircraft $b^*$ for $n_r\neq n_r^*$
 on the desired signal, and the third term is the interference imposed by aircraft $a$ for
 $1\le a\le A$ on the desired signal.

\section{Analysis of Achievable Throughput of RZF-TPC}\label{S3}

 Since $b^*$ does not know the estimated CSI, the achievable ergodic rate is
 adopted. We will also take into account the channel estimation error. From the signal
 $Y^{b^*}_{n_r^{*}}$ (\ref{eq14}) received at the DRA $n_r^*$ of $b^*$, the power of the desired 
 signal $P_{{\rm S}_{b^*,n_r^*}^{a^*}}$ and the power of the interference pulse noise
 $P_{{\rm I\&N}_{b^*,n_r^*}^{a^*}}$ can be obtained respectively as
\begin{align}
 P_{{\rm S}_{b^*,n_r^*}^{a^*}} =&  P_{r,b^*}^{a^*} \left|\mathcal{E}\left\{
  \left[\bm{H}_{b^*}^{a^*}\right]_{[n_r^{*}:~]}
  \left[\bm{V}_{b^*}^{a^*}\right]_{[~:n_r^{*}]}\right\}\right|^2 , \label{eq15} \\
 P_{{\rm I\&N}_{b^*,n_r^*}^{a^*}} =& P_{r,b^*}^{a^*} \text{Var}\left\{
  \left[\bm{H}_{b^*}^{a^*}\right]_{[n_r^{*}:~]}
  \left[\bm{V}_{b^*}^{a^*}\right]_{[~:n_r^{*}]}\right\} + P_{r,b^*}^{a^*}
  \sum\limits_{n_r\neq n_r^*}\mathcal{E}\left\{\left| 
  \left[\bm{H}_{b^*}^{a^*}\right]_{[n_r^{*}:~]} 
  \left[\bm{V}_{b^*}^{a^*}\right]_{[~:n_r]}\right|^{2}\right\} \nonumber \\
 & + \sum\limits_{a=1}^A P_{r,b^*}^a \sum\limits_{n_r=1}^{N_r}\mathcal{E}\left\{\left|
 \left[\bm{H}_{b^*}^a\right]_{[n_r^{*}:~]}\left[\bm{V}_{b^a}^a\right]_{[~:n_r]}\right|^{2}\right\}
  + \sigma_w^2  , \label{eq16}
\end{align}
 where $\text{Var}\left\{~\right\}$ is the variance operator. Thus, the SINR
 at $n_r^*$-th DRA of $b^*$ is given by 
\begin{align}\label{eq17}
 \gamma_{b^*,n_r^*}^{a^*} =& \frac{ P_{{\rm S}_{b^*,n_r^*}^{a^*}} } 
  { P_{{\rm I\&N}_{b^*,n_r^*}^{a^*}} } ,
\end{align}
 and the achievable transmission rate per antenna between the transmitting aircraft $a^*$ and
 the destination aircraft $b^*$ can be readily expressed as
\begin{align}\label{eq18}
 C_{b^*}^{a^*} =& \frac{1}{N_r} \sum\limits_{n_r^*=1}^{N_r} \log_2 \left(1 + \gamma_{b^*,n_r^*}^{a^*}\right) .
\end{align}

\subsection{Statistics of Channel Estimate}\label{S3.1}

 The MMSE channel estimate $\left[\widehat{\bm{H}}_{b^*}^{a^*}\right]_{[n_r:~]}$ is
 related to the true channel $\left[\bm{H}_{b^*}^{a^*}\right]_{[n_r:~]}$ by
\begin{align}\label{eq19}
 \left[\bm{H}_{b^*}^{a^*}\right]_{[n_r:~]} =& \left[\widehat{\bm{H}}_{b^*}^{a^*}\right]_{[n_r:~]}
  + \left[\widetilde{\bm{H}}_{b^*}^{a^*}\right]_{[n_r:~]} ,
\end{align}
 where the estimation error $\left[\widetilde{\bm{H}}_{b^*}^{a^*}\right]_{[n_r:~]}$ is
 statistically independent of both $\left[\widehat{\bm{H}}_{b^*}^{a^*}\right]_{[n_r:~]}$
 and $\left[\bm{H}_{b^*}^{a^*}\right]_{[n_r:~]}$ \cite{kay2003fundamentals}. Recalling
 the distribution (\ref{eq8}), we have
\begin{align}\label{eq20}
  & \bm{vec}\left(\widehat{\bm{H}}_{b^{*}}^{a^{*}}\right)\sim
 \mathcal{CN}\left(\bm{vec}\left(\nu\bm{H}_{\text{d},b^{*}}^{a^*}\right),
 \bm{\Phi}_{b^{*}}^{a^{*}}\right) ,
\end{align}
 where $\bm{\Phi}_{b^{*}}^{a^{*}}$ is the covariance matrix  of the MMSE estimate
 $\bm{vec}\left(\widehat{\bm{H}}_{b^{*}}^{a^{*}}\right)$ given by
\begin{align}\label{eq21}
  \bm{\Phi}_{b^*}^{a^*} =& \varsigma^2 \bar{\bm{R}}_{\text{r},r}^t \left(
  \frac{\sigma_w^2}{P_{r,b^{*}}^{a^{*}}} \bm{I}_{N_r N_t} \! +\! \varsigma^2 \bar{\bm{R}}_{\text{r},r}^t
  \! + \! \sum\limits_{a=1}^A \frac{P_{r,a^{*}}^a}{P_{r,b^{*}}^{a^{*}}} \varsigma^2
  \bar{\bm{R}}_{\text{r},r}^t\right)^{-1}  \varsigma^2 \bar{\bm{R}}_{\text{r},r}^t .
\end{align}
 The spatial correlation matrix $\bar{\bm{R}}_{{\rm r},r}^t$ in (\ref{eq21}) is given by
 $\bar{\bm{R}}_{{\rm r},r}^t=\bar{\bm{R}}_t \otimes \bar{\bm{R}}^r$, and we have
 $\bar{\bm{R}}^r=\bm{I}_{N_r}$.

 The distribution of $\bm{vec}\left(\widetilde{\bm{H}}_{b^*}^{a^*}\right)$ is given by
\begin{align}\label{eq22}
 & \bm{vec}\left(\widetilde{\bm{H}}_{b^*}^{a^*}\right) \sim \mathcal{CN}\left(\bm{0}_{N_tN_r},
  \bm{\Xi}_{b^*}^{a^*}\right) ,
\end{align}
 whose covariance matrix $\bm{\Xi}_{b^*}^{a^*}$ can be expressed as
\begin{align}\label{eq23}
 \bm{\Xi}_{b^*}^{a^*} =& \varsigma^2\bar{\bm{R}}_{{\rm r},r}^t - \bm{\Phi}_{b^*}^{a^*}
  = \left[\begin{array}{cccc} \left[\bm{\Xi}_{b^*}^{a^*}\right]_{(1,1)} &
  \left[\bm{\Xi}_{b^*}^{a^*}\right]_{(1,2)} & \cdots & \left[\bm{\Xi}_{b^*}^{a^*}\right]_{(1,N_r)} \\
  \vdots & \vdots & \cdots & \vdots \\
  \left[\bm{\Xi}_{b^*}^{a^*}\right]_{(N_r,1)} & \left[\bm{\Xi}_{b^*}^{a^*}\right]_{(N_r,2)} & \cdots
  & \left[\bm{\Xi}_{b^*}^{a^*}\right]_{(N_r,N_r)} \end{array}\right]
  \in \mathbb{C}^{N_tN_r\times N_tN_r} ,
\end{align}
 where $\left[\bm{\Xi}_{b^*}^{a^*}\right]_{(i,j)}=\mathcal{E}\left\{
 \left[\widetilde{\bm{H}}_{b^*}^{a^*}\right]_{[i:~]}^{\rm H}
 \left[\widetilde{\bm{H}}_{b^*}^{a^*}\right]_{[j:~]}\right\}\in
 \mathbb{C}^{N_t\times N_t}$, $\forall i,j\in \{1,2,\cdots, N_r\}$. This indicates
 that the distribution of $\left[\widetilde{\bm{H}}_{b^*}^{a^*}\right]_{[n_r:~]}$
 is given by
\begin{align}\label{eq24}
 \left[\widetilde{\bm{H}}_{b^*}^{a^*}\right]_{[n_r:~]}^{\rm T} \sim &
 \mathcal{CN}\left(\bm{0}_{N_t},\left[\bm{\Xi}_{b^*}^{a^*}\right]_{(n_r,n_r)}\right) .
\end{align}

 Furthermore, the correlation matrix $\mathcal{E}\left\{\bm{vec}\left(
 \widehat{\bm{H}}_{a^*}^{b^*}\right) \bm{vec}\left(\widehat{\bm{H}}_{a^*}^{b^*}\right)^{\rm H}\right\}
 =\nu^2\bm{M}_{b^*}^{a^*}+\bm{\Phi}_{b^*}^{a^*}$, where
\begin{align}\label{eq25}
 \bm{M}_{b^*}^{a^*} =& \bm{vec}\left(\bm{H}_{\text{d},a^*}^{b^*}\right)
  \bm{vec}\left(\bm{H}_{{\rm d},a^*}^{b^*}\right)^{\rm H} \in \mathbb{C}^{N_tN_r\times N_tN_r} .
\end{align}
 $\bm{M}_{b^*}^{a^*}$ can be expressed in a form similar to (\ref{eq23}) having the
 $(i,j)$-th sub-matrix of $\left[\bm{M}_{b^*}^{a^*}\right]_{(i,j)}=\left[
 \bm{H}_{{\rm d},a^*}^{b^*}\right]_{[i:~]}^{\rm H}\left[\bm{H}_{{\rm d},a^*}^{b^*}\right]_{[j:~]}
 \in\mathbb{C}^{N_t\times N_t}$, $\forall i,j\in \{1,2,\cdots, N_r\}$. Likewise,
 $\bm{\Phi}_{b^*}^{a^*}$ has a form similar to that of (\ref{eq23}) having the $(i,j)$th
 sub-matrix of  $\left[\bm{\Phi}_{b^*}^{a^*}\right]_{(i,j)}\in\mathbb{C}^{N_t\times N_t}$,
 $\forall i,j\in \{1,2,\cdots, N_r\}$.

\subsection{Desired Signal Power}\label{S3.2}

 Four useful lemmas are collected in Appendix~\ref{Apa}. In order to exploit Lemma~\ref{L1}
 for calculating the desired signal power, we define
\begin{align}\label{eq26}
 \bm{\Upsilon}_{b^*,\emptyset n_r^*}^{a^*} =& \left(\frac{1}{N_t}\big(\widehat{\bm{H}}_{b^*}^{a^*}\big)^{\rm H}
  \widehat{\bm{H}}_{b^*}^{a^*} - \frac{1}{N_t}\left[\widehat{\bm{H}}_{b^*}^{a^*}\right]_{[n_{r}^*:~]}^{\rm H}
  \left[\widehat{\bm{H}}_{b^*}^{a^*}\right]_{[n_{r}^*:~]} + \xi_{b^*}^{a^*}\bm{I}_{N_{t}}\right)^{-1} .
\end{align}
 Clearly, $\bm{\Upsilon}_{b^*,\emptyset n_r^*}^{a^*}$ is independent of
 $\left[\widehat{\bm{H}}_{b^*}^{a^*}\right]_{[n_{r}^*:~]}$. Recalling $\bm{\Upsilon}_{b^*}^{a^*}$
 of (\ref{eq12}), we can express
 $\bm{\Upsilon}_{b^*}^{a^*}\left[\widehat{\bm{H}}_{b^*}^{a^*}\right]_{[n_{r}^{*}:~]}^{\rm H}$ as 
\begin{align}\label{eq27}
 \bm{\Upsilon}_{b^*}^{a^*}\left[\widehat{\bm{H}}_{b^*}^{a^*}\right]_{[n_{r}^{*}:~]}^{\rm H} 
  =& \frac{\bm{\Upsilon}_{b^*,\emptyset n_r^{*}}^{a^*}
  \left[\widehat{\bm{H}}_{b^*}^{a^*}\right]_{[n_{r}^{*}:~]}^{\rm H}}{1 +
  \frac{1}{N_{t}}\left[\widehat{\bm{H}}_{b^*}^{a^*}\right]_{[n_{r}^{*}:~]}
  \bm{\Upsilon}_{b^*,\emptyset n_r^{*}}^{a^*}\left[\widehat{\bm{H}}_{b^*}^{a^*}\right]_{[n_{r}^{*}:~]}^{\rm H}} ,
\end{align}
 according to Lemma~\ref{L1}. Furthermore, $\left[\bm{H}_{b^*}^{a^*}\right]_{[n_r^*:~]}
 \left[\bm{V}_{b^*}^{a^*}\right]_{[~:n_r^*]}$ can be formulated as
\begin{align}\label{eq28}
 \left[\bm{H}_{b^*}^{a^*}\right]_{[n_r^{*}:~]}  \left[\bm{V}_{b^*}^{a^*}\right]_{[~:n_r^{*}]} =&
  \frac{\left[\bm{H}_{b^*}^{a^*}\right]_{[n_r^{*}:~]}\bm{\Upsilon}_{b^*,\emptyset n_r^{*}}^{a^*}
  \left[\widehat{\bm{H}}_{b^*}^{a^*}\right]_{[n_{r}^{*}:~]}^{\rm H}}{1 + \frac{1}{N_t}
  \left[\widehat{\bm{H}}_{b^*}^{a^*}\right]_{[n_{r}^{*}:~]}\bm{\Upsilon}_{b^*,\emptyset n_r^{*}}^{a^*}
  \left[\widehat{\bm{H}}_{b^*}^{a^*}\right]_{[n_{r}^{*}:~]}^{\rm H}} .
\end{align}
 Recalling Lemmas~\ref{L2} to \ref{L4} and (\ref{eq19}) as well as the fact that
 $\left[\widetilde{\bm{H}}_{b^*}^{a^*}\right]_{[n_r^{*}:~]}$ is independent of
 $\left[\widehat{\bm{H}}_{b^*}^{a^*}\right]_{[n_r^{*}:~]}$, the expectation of
 $\left[\bm{H}_{b^*}^{a^*}\right]_{[n_r^{*}:~]}\bm{\Upsilon}_{b^*,\emptyset n_r^{*}}^{a^*}
 \left[\widehat{\bm{H}}_{b^*}^{a^*}\right]_{[n_{r}^{*}:~]}^{\rm H}$ can be rewritten as
\begin{align}\label{eq29}
 \mathcal{E}\left\{\left[\bm{H}_{b^*}^{a^*}\right]_{[n_r^{*}:~]}\bm{\Upsilon}_{b^*,\emptyset n_r^{*}}^{a^*}
  \left[\widehat{\bm{H}}_{b^*}^{a^*}\right]_{[n_{r}^{*}:~]}^{\rm H} \right\} =&
  \mathcal{E}\left\{ \left(\left[\widehat{\bm{H}}_{b^*}^{a^*}\right]_{[n_r^{*}:~]} +
  \left[\widetilde{\bm{H}}_{b^*}^{a^*}\right]_{[n_r^{*}:~]}\right)\bm{\Upsilon}_{b^*,\emptyset n_r^{*}}^{a^*}
  \left[\widehat{\bm{H}}_{b^*}^{a^*}\right]_{[n_{r}^{*}:~]}^{\rm H}\right\} \nonumber \\
 =& \vartheta_{b^*,n_{r}^{*}}^{a^*} =  \text{Tr}\left\{\left[\bm{\Theta}_{b^{*}}^{a^{*}}\right]_{(n_r^{*},n_r^{*})}\bm{\Upsilon}_{b^*, \emptyset n_r^{*}}^{a^*}\right\} ,
\end{align}
 in which $\text{tr}\{\cdot \}$ denotes the matrix-trace operation, and
\begin{align}
 \vartheta_{b^*,n_{r}^{*}}^{a^*} =&
  \left[\widehat{\bm{H}}_{b^*}^{a^*}\right]_{[n_{r}^{*}:~]}\bm{\Upsilon}_{b^*,\emptyset n_r^{*}}^{a^*}
  \left[\widehat{\bm{H}}_{b^*}^{a^*}\right]_{[n_{r}^{*}:~]}^{\rm H} , \label{eq30} \\
 \left[\bm{\Theta}_{b^{*}}^{a^{*}}\right]_{(n_r^{*},n_r^{*})} =& \nu^2\left[\bm{M}_{b^{*}}^{a^{*}}\right]_{(n_r^{*},n_r^{*})}
  + \left[\bm{\Phi}_{b^{*}}^{a^{*}}\right]_{(n_r^{*},n_r^{*})} \label{eq31} .
\end{align}

 The following theorem is required for the asymptotic analysis of the achievable data rate.

\begin{theorem}[Deterministic equivalents \cite{hachem2007deterministic,hachem2012clt}]\label{T1}
 Let $\bm{H}=\nu\bar{\bm{H}}+\varsigma (\bm{R})^{\frac{1}{2}} \frac{\bm{G}}{\sqrt{N}}
 (\widetilde{\bm{R}})^{\frac{1}{2}}\in \mathbb{C}^{N\times K}$, where $\bar{\bm{H}}\in
 \mathbb{C}^{N\times K}$ is a deterministic matrix, $\bm{R}\in \mathbb{C}^{N\times N}$ and
 $\widetilde{\bm{R}}\in \mathbb{C}^{K\times K}$ are deterministic diagonal matrices with
 non-negative diagonal elements, and $\bm{G}\in \mathbb{C}^{N\times K}$ is a random matrix
 with each element obeying the distribution $\mathcal{CN}(0,1)$, while $\nu\in [0, ~1]$ and
 $\varsigma\in [0,~1]$ with $\nu^{2}+\varsigma^{2}=1$ are the weighting factors of
 $\bar{\bm{H}}$ and $(\bm{R})^{\frac{1}{2}} \frac{\bm{G}}{\sqrt{N}}
 (\widetilde{\bm{R}})^{\frac{1}{2}}$, respectively. Furthermore, $\bm{R}$ and $\widetilde{\bm{R}}$
 have uniformly bounded spectral norms with respect to $K$ and $N$. The matrices
\begin{align}
 \bm{T}(\xi ) =& \left(\xi\left(\bm{I}_N + \tilde{\delta}\varsigma^2 \bm{R}\right) + \nu^2
  \bar{\bm{H}}\left(\bm{I}_K + \delta\widetilde{\bm{R}}\right)^{-1} \bar{\bm{H}}^{\rm H}\right)^{-1} ,
  \label{eq32} \\
 \widetilde{\bm{T}}(\xi ) =& \left(\xi\left(\bm{I}_K + \delta\varsigma^2 \widetilde{\bm{R}}\right) + 
  \nu^2 \bar{\bm{H}}^{\rm H}\left(\bm{I}_N + \tilde{\delta}\bm{R}\right)^{-1}\bar{\bm{H}}\right)^{-1} ,
  \label{eq33}
\end{align}
 are the respective approximations of the resolvent $\bm{Q}(\xi)$ and the co-resolvent
 $\widetilde{\bm{Q}}(\xi)$
\begin{align}
 \bm{Q}(\xi ) =& \left(\bm{H} \bm{H}^{\rm H} + \xi\bm{I}_N\right)^{-1} , \label{eq34} \\
 \widetilde{\bm{Q}}(\xi ) =& \left(\bm{H}^{\rm H}\bm{H} + \xi\bm{I}_K\right)^{-1} . \label{eq35}
\end{align}
 In (\ref{eq32}) and  (\ref{eq33}), $(\delta, \tilde{\delta})$ admits a unique solution in
 the class of Stieltjes transforms \cite{ismail1979special} of non-negative measures with
 the support in $\mathbb{R}^{+}$, which are given by
\begin{align}
 \delta =& \frac{1}{N}\text{Tr}\left\{\varsigma^2 \bm{R}\left(\xi\left(\bm{I}_N +
  \tilde{\delta}\varsigma^{2} \bm{R}\right) + \nu^2 \bar{\bm{H}}\left(\bm{I}_K +
  \delta\varsigma^2 \widetilde{\bm{R}}\right)^{-1}\bar{\bm{H}}^{\rm H}\right)^{-1}\right\} , \label{eq36} \\
 \tilde{\delta} =& \frac{1}{N}\text{Tr}\left\{\varsigma^2 \widetilde{\bm{R}}\left(\xi\left(\bm{I}_K
  + \delta\varsigma^2 \widetilde{\bm{R}}\right) + \nu^2 \bar{\bm{H}}^{\rm H}\left(\bm{I}_N +
  \tilde{\delta}\varsigma^2 \bm{R}\right)^{-1}\bar{\bm{H}}\right)^{-1}\right\} . \label{eq37}
\end{align}
 Then $\delta$ and $\tilde{\delta}$ can be numerically solved as
\begin{align}\label{eq38}
 \delta =& \lim\limits_{t\to \infty} \delta^{(t)} \text{ and }
 \tilde{\delta} = \lim\limits_{t\to \infty} \tilde{\delta}^{(t)} ,
\end{align}
 by defining $\delta^{(t)}$ and $\tilde{\delta}^{(t)}$
\begin{align}
 \delta^{(t)} =& \frac{1}{N}\text{Tr}\left\{\varsigma^2 \bm{R}\left(\xi\left(\bm{I}_N
  + \tilde{\delta}^{(t-1)}\varsigma^2 \bm{R}\right) + \nu^2 \bar{\bm{H}}\left(\bm{I}_K
  + \delta^{(t-1)}\varsigma^2 \widetilde{\bm{R}}\right)^{-1} \bar{\bm{H}}^{\rm H}\right)^{-1}\right\} ,\label{eq39} \\
 \tilde{\delta}^{(t)} =& \frac{1}{N}\text{Tr}\left\{\varsigma^2 \widetilde{\bm{R}}\left(\xi\left(\bm{I}_K
  + \delta^{(t)}\varsigma^2 \widetilde{\bm{R}}\right) + \nu^2 \bar{\bm{H}}^{\rm H}\left(\bm{I}_N
  + \tilde{\delta}^{(t-1)}\varsigma^2 \bm{R}\right)^{-1}\bar{\bm{H}}\right)^{-1}\right\} , \label{eq40}
\end{align}
 with the initial values of $\delta^{(0)}=\tilde{\delta}^{(0)}=\frac{1}{\xi}$.
\end{theorem}
 
 According to Theorem~\ref{T1}, $\bm{\Upsilon}_{b^*}^{a^*}$ can be approximated as
\begin{align}\label{eq41}
 \bm{\Upsilon}_{b^*}^{a^*} \approx & \bigg(\xi_{b^*}^{a^*}\left(\bm{I}_{N_t} + \tilde{\delta}_{b^*}^{a^*}
  \varsigma^2 \bm{\Phi}_{a^*}\right) + \frac{1}{N_t}\nu^2\big(\bm{H}_{{\rm d},b^*}^{a^*}\big)^{\rm H}
  \left(\bm{I}_{N_r} + \delta_{b^*}^{a^*} \varsigma^2 \bm{\Phi}^{b^*}\right)^{-1}\bm{H}_{{\rm d},b^*}^{a^*}
  \bigg)^{-1} .
\end{align}
 By recalling (\ref{eq36}) and (\ref{eq37}), $\delta_{b^*}^{a^*}$ and $\tilde{\delta}_{b^*}^{a^*}$
 in (\ref{eq41}) are given by
\begin{align}
 \delta_{b^*}^{a^*}\! =& \frac{1}{N_t}\text{Tr}\Bigg\{\! \varsigma^2 \bm{\Phi}_{a^*}
  \bigg(\xi_{b^*}^{a^*}\left(\bm{I}_{N_t}\! +\! \tilde{\delta}_{b^*}^{a^*}\varsigma^2 \bm{\Phi}_{a^*}\right)\! +\! 
  \frac{1}{N_t}\nu^2 \big(\bm{H}_{{\rm d},b^*}^{a^*}\big)^{\rm H}\left(\bm{I}_{N_r}\! +\!
  \delta_{b^*}^{a^*}\varsigma^2 \bm{\Phi}^{b^*}\right)^{-1}\bm{H}_{{\rm d},b^*}^{a^*} \bigg)^{-1}\! \Bigg\}\! ,\! 
  \label{eq42} \\
 \tilde{\delta}_{b^*}^{a^*}\! =& \frac{1}{N_t}\text{Tr}\Bigg\{\! \varsigma^2 \bm{\Phi}^{b^*}
  \bigg(\xi_{b^*}^{a^*}\left(\bm{I}_{N_r}\! +\! \delta_{b^*}^{a^*}\varsigma^2 \bm{\Phi}^{b^*}\right)\! + \!
  \frac{1}{N_t}\nu^2 \bm{H}_{{\rm d},b^*}^{a^*}\left(\bm{I}_{N_t}\! +\! \tilde{\delta}_{b^*}^{a^*}\varsigma^2
  \bm{\Phi}_{a^*}\right)^{-1}\big(\bm{H}_{{\rm d},b^*}^{a^*}\big)^{\rm H} \bigg)^{-1} \! \Bigg\}\! .\! \label{eq43}
\end{align}
 Furthermore, given (\ref{eq41}) and recalling (\ref{eq31}), $\bm{\Upsilon}_{b^*,\emptyset n_r^{*}}^{a^*}$
 can be approximated as
\begin{align}\label{eq44}
 \bm{\Upsilon}_{b^*,\emptyset n_r^*}^{a^*}\! \approx & \bigg(\! \xi_{b^*}^{a^*}\! \left(\! \bm{I}_{N_t}\! +\!
  \tilde{\delta}_{b^*}^{a^*}\varsigma^2 \bm{\Phi}_{a^*}\! \right)\!\! +\! \frac{1}{N_t}\nu^2
  \big(\bm{H}_{{\rm d},b^*}^{a^*}\big)^{\rm H}\! \left(\! \bm{I}_{N_r}\! +\! \delta_{b^*}^{a^*} \varsigma^2
  \bm{\Phi}^{b^*}\right)^{-1}\bm{H}_{{\rm d},b^*}^{a^*} 
 \! -\!\frac{1}{N_t}\left[\bm{\Theta}_{b^*}^{a^*}\right]_{(n_r^*,n_r^*)} \! \bigg)^{-1} \!\! .
\end{align}
 Then $\vartheta_{b^*,n_r^*}^{a^*}$ of (\ref{eq30}) can be further expressed as
\begin{align}\label{eq45}
 \vartheta_{b^*,n_r^*}^{a^*} =& \text{Tr}\bigg\{\left[\bm{\Theta}_{b^*}^{a^*}\right]_{(n_r^*,n_r^*)}
  \bigg(\xi_{b^*}^{a^*}\left(\bm{I}_{N_t} + \tilde{\delta}_{b^*}^{a^*}\varsigma^2 \bm{\Phi}_{a^*}\right)  \nonumber \\
 & +\frac{1}{N_t}\nu^2\big(\bm{H}_{{\rm d},b^*}^{a^*}\big)^{\rm H}\left(\bm{I}_{N_r} +
  \delta_{b^*}^{a^*} \varsigma^2 \bm{\Phi}^{b^*}\right)^{-1}\bm{H}_{{\rm d},b^*}^{a^*} -
  \frac{1}{N_t}\left[\bm{\Theta}_{b^*}^{a^*}\right]_{(n_r^*,n_r^*)} \bigg)^{-1}	\bigg\} .
\end{align}
	
 Hence, noting (\ref{eq28}), (\ref{eq29}) and (\ref{eq45}), the desired signal power of
 (\ref{eq15}) can be expressed as
\begin{align}\label{eq46}
 P_{{\rm S}_{b^*,n_r^*}^{a^*}} =& P_{r,b^*}^{a^*}\bigg(\frac{\vartheta_{b^*,n_r^*}^{a^*}}{1 + \frac{1}{N_t}
  \vartheta_{b^*,n_r^*}^{a^*}}\bigg)^2 .
\end{align}

\subsection{Interference Plus Noise Power}\label{S3.3}

 Recalling (\ref{eq30}) and Lemma~\ref{L1}, we can express $\big[\widehat{\bm{H}}_{b^*}^{a^*}\big]_{[n_r^*:~]}
 \big[\bm{V}_{b^*}^{a^*}\big]_{[~:n_r^*]}$ as
\begin{align}\label{eq47}
 \left[\widehat{\bm{H}}_{b^*}^{a^*}\right]_{[n_r^*:~]}\left[\bm{V}_{b^*}^{a^*}\right]_{[~:n_r^*]} =&
  \frac{\vartheta_{b^*,n_{r}^{*}}^{a^*}}{1 + \frac{1}{N_t}\vartheta_{b^*,n_r^*}^{a^*}} .
\end{align}
 Thus, we have
\begin{align}\label{eq48}
 \text{Var}\left\{\big[\bm{H}_{b^*}^{a^*}\big]_{[n_r^*:~]}\big[\bm{V}_{b^*}^{a^*}\big]_{[~:n_r^*]}\right\}
  =& \mathcal{E}\Bigg\{\Bigg| \frac{\big[\widetilde{\bm{H}}_{b^*}^{a^*}\big]_{[n_r^*:~]}
  \bm{\Upsilon}_{b^*,\emptyset n_r^{*}}^{a^*}\big[\widehat{\bm{H}}_{b^*}^{a^*}\big]_{[n_r^{*}:~]}^{\rm H}}
  {1 + \frac{1}{N_t}\big[\widehat{\bm{H}}_{b^*}^{a^*}\big]_{[n_r^*:~]}
  \bm{\Upsilon}_{b^*,\emptyset n_r^{*}}^{a^*}\big[\widehat{\bm{H}}_{b^*}^{a^*}\big]_{[n_r^*:~]}^{\rm H}}
  \Bigg|^2\Bigg\} \nonumber \\
 =& \text{Tr}\Bigg\{ \frac{\bm{\Upsilon}_{b^*,\emptyset n_r^*}^{a^*}
  \big[\bm{\Theta}_{b^*}^{a^*}\big]_{(n_r^*,n_r^*)}\bm{\Upsilon}_{b^*,\emptyset n_r^*}^{a^*}
  \big[\bm{\Xi}_{b^*}^{a^*}\big]_{(n_r^*,n_r^*)}}{\left(1 + \frac{1}{N_t}
  \vartheta_{b^*,n_r^*}^{a^*}\right)^2} \Bigg\} .
\end{align}

 From Lemmas~\ref{L1} and \ref{L3}, $\mathcal{E}\left\{\big| \big[\bm{H}_{b^*}^{a^*}
 \big]_{[n_r^*:~]} \big[\bm{V}_{b^*}^{a^*}\big]_{[~:n_r]}\big|^2\right\}$, $n_r\neq n_r^*$,
 can be expressed as
\begin{align}\label{eq49}
 \mathcal{E}\left\{\left| \left[\bm{H}_{b^*}^{a^*}\right]_{[n_r^*:~]}
  \left[\bm{V}_{b^*}^{a^*}\right]_{[~:n_r]}\right|^2\right\} &= \mathcal{E}\Bigg\{
  \Bigg|\frac{\big[\bm{H}_{b^*}^{a^*}\big]_{[n_r^*:~]}\bm{\Upsilon}_{b^*,\emptyset n_r}^{a^*}
  \big[\widehat{\bm{H}}_{b^*}^{a^*}\big]_{[n_r:~]}^{\rm H}}{1 + \frac{1}{N_t}
  \big[\widehat{\bm{H}}_{b^*}^{a^*}\big]_{[n_r:~]}\bm{\Upsilon}_{a^*,\emptyset n_r}^{b^*}
  \big[\widehat{\bm{H}}_{b^*}^{a^*}\big]_{[n_r:~]}^{\rm H}}\Bigg|^2\Bigg\} , \nonumber \\
 & \hspace*{-10mm}= \mathcal{E}\Bigg\{ \frac{\left[\bm{H}_{b^*}^{a^*}\right]_{[n_r^*:~]}
  \bm{\Upsilon}_{b^*,\emptyset n_r}^{a^*}\left[\bm{\Theta}_{b^*}^{a^*}\right]_{(n_r,n_r)}
  \bm{\Upsilon}_{b^*,\emptyset n_r}^{a^*}\left[\bm{H}_{b^*}^{a^*}\right]_{[n_r^*:~]}^{\rm H}}
  {(1 + \frac{1}{N_t}\vartheta_{b^*,n_r}^{a^*})^2}\Bigg\} ,
\end{align}
 in which
\begin{align}
 \vartheta_{b^*,n_r}^{a^*} =& \text{Tr}\left\{\left[\bm{\Theta}_{b^*}^{a^*}\right]_{(n_r,n_r)}
  \bm{\Upsilon}_{b^*,\emptyset n_r}^{a^*}\right\} , \label{eq50} \\
 \bm{\Upsilon}_{b^*,\emptyset n_r}^{a^*}\!\! =& \bigg(\! \xi_{b^*}^{a^*}\!\! \left(\!\! \bm{I}_{N_t}\!\! +\!
  \frac{1}{N_t}\tilde{\delta}_{b^*}^{a^*}\varsigma^2 \bm{\Phi}_{a^*}\!\! \right)\!\! +\! \frac{1}{N_t}\nu^2
  \big(\bm{H}_{{\rm d},b^*}^{a^*}\big)^{\rm H}\! \left(\bm{I}_{N_r}\!\! +\! \delta_{b^*}^{a^*} \varsigma^2
  \bm{\Phi}^{b^*}\right)^{-1}\! \bm{H}_{{\rm d},b^*}^{a^*}  
  \!\! -\! \frac{1}{N_t}\! \left[\bm{\Theta}_{b^*}^{a^*}\right]_{(n_r,n_r)} \!\! \bigg)^{-1} \!\!\! .\!
  \label{eq51}
\end{align}
 According to Lemma~\ref{L2}, we have
\begin{align}\label{eq52}
 \bm{\Upsilon}_{b^*,\emptyset n_r}^{a^*} =& \bm{\Upsilon}_{b^*,\emptyset n_r,n_r^*}^{a^*} -
  \frac{\frac{1}{N_t}\bm{\Upsilon}_{b^*,\emptyset n_r,n_r^*}^{a^*}
  \left[\widehat{\bm{H}}_{b^*}^{a^*}\right]_{[n_r^*:~]}^{\rm H} \left[\widehat{\bm{H}}_{b^*}^{a^*}\right]_{[n_r^*:~]}
  \bm{\Upsilon}_{b^*,\emptyset n_r,n_r^*}^{a^*}}{1 + \frac{1}{N_t}\left[\widehat{\bm{H}}_{b^*}^{a^*}\right]_{[n_r^*:~]}
  \bm{\Upsilon}_{b^*,\emptyset n_r,n_r^*}^{a^*}\left[\widehat{\bm{H}}_{b^*}^{a^*}\right]_{[n_r^*:~]}^{\rm H}} ,
\end{align}
 where $\bm{\Upsilon}_{b^*,\emptyset n_r,n_r^*}^{a^*}$ is independent of
 $\left[\widehat{\bm{H}}_{b^*}^{a^*}\right]_{[n_r^*:~]}$ and
 $\left[\widehat{\bm{H}}_{b^*}^{a^*}\right]_{[n_r:~]}$, and it can be approximated as
\begin{align}\label{eq53}
 \bm{\Upsilon}_{b^*,\emptyset n_r, n_r^*}^{a^*} =& \bigg(\xi_{b^*}^{a^*}\left(\bm{I}_{N_t} +
  \tilde{\delta}_{b^*}^{a^*}\varsigma^2 \bm{\Phi}_{a^*}\right) + \frac{1}{N_t}\nu^2
  \big(\bm{H}_{{\rm d},b^*}^{a^*}\big)^{\rm H}\left(\bm{I}_{N_r} + \delta_{b^*}^{a^*}
  \varsigma^2 \bm{\Phi}^{b^*}\right)^{-1}\bm{H}_{{\rm d},b^*}^{a^*} \nonumber \\
 & -\frac{1}{N_t}\left(\left[\bm{\Theta}_{b^*}^{a^*}\right]_{(n_r,n_r)} +
  \left[\bm{\Theta}_{b^*}^{a^*}\right]_{(n_r^*,n_r^*)}\right)\bigg)^{-1} .
\end{align}
 Then, we can rewrite $\left[\bm{H}_{b^*}^{a^*}\right]_{[n_r^*:~]}\bm{\Upsilon}_{b^*,\emptyset n_r}^{a^*}
 \left[\bm{\Theta}_{b^*}^{a^*}\right]_{(n_r,n_r)}\bm{\Upsilon}_{b^*,\emptyset n_r}^{a^*}
\left[\bm{H}_{b^*}^{a^*}\right]_{[n_r^*:~]}^{\rm H}$ as
\begin{align}\label{eq54}
 & \left[\bm{H}_{b^*}^{a^*}\right]_{[n_r^*:~]}\bm{\Upsilon}_{b^*,\emptyset n_r}^{a^*}
  \left[\bm{\Theta}_{b^*}^{a^*}\right]_{(n_r,n_r)}
  \bm{\Upsilon}_{b^*,\emptyset n_r}^{a^*}\left[\bm{H}_{b^*}^{a^*}\right]_{[n_r^*:~]}^{\rm H} = \nonumber \\
 & \hspace*{10mm}\left[\bm{H}_{b^*}^{a^*}\right]_{[n_r^*:~]}\bm{\Upsilon}_{b^*,\emptyset n_r,n_r^*}^{a^*}
  \left[\bm{\Theta}_{b^*}^{a^*}\right]_{(n_r,n_r)}\bm{\Upsilon}_{b^*,\emptyset n_r,n_r^*}^{a^*}
  \left[\bm{H}_{b^*}^{a^*}\right]_{[n_r^*:~]}^{\rm H} \nonumber \\
 & \hspace*{10mm}- 2\Re\left\{\frac{\frac{1}{N_t}\big[\bm{H}_{b^*}^{a^*}\big]_{[n_r^*:~]}
  \bm{\Upsilon}_{b^*,\emptyset n_r,n_r^*}^{a^*}\big[\widehat{\bm{H}}_{b^*}^{a^*}\big]_{[n_r^*:~]}^{\rm H}
  \big[\widehat{\bm{H}}_{b^*}^{a^*}\big]_{[n_r^*:~]} \widetilde{\bm{\Upsilon}}_{a^*,\emptyset n_r,n_r^*}^{b^*}
  \big[\bm{H}_{b^*}^{a^*}\big]_{[n_r^*:~]}^{\rm H}}{1 + \frac{1}{N_t}
  \big[\widehat{\bm{H}}_{b^*}^{a^*}\big]_{[n_r^*:~]}\bm{\Upsilon}_{b^*,\emptyset n_r,n_r^*}^{a^*}
  \big[\widehat{\bm{H}}_{b^*}^{a^*}\big]_{[n_r^*:~]}^{\rm H}} \right\} \nonumber \\
 & \hspace*{10mm}+ \frac{\frac{1}{N_t^2}\left|\big[\bm{H}_{b^*}^{a^*}\big]_{[n_r^*:~]}
  \bm{\Upsilon}_{b^*,\emptyset n_r,n_r^*}^{a^*}\big[\widehat{\bm{H}}_{b^*}^{a^*}\big]_{[n_r^*:~]}^{\rm H}\right|^2
  \big[\widehat{\bm{H}}_{b^*}^{a^*}\big]_{[n_r^*:~]} \widetilde{\bm{\Upsilon}}_{b^*,\emptyset n_r,n_r^*}^{a^*}
  \big[\widehat{\bm{H}}_{b^*}^{a^*}\big]_{[n_r^*:~]}^{\rm H} }{\left(1 + \frac{1}{N_t}
  \big[\widehat{\bm{H}}_{b^*}^{a^*}\big]_{[n_r^*:~]}\bm{\Upsilon}_{b^*,\emptyset n_r,n_r^*}^{a^*}
  \big[\widehat{\bm{H}}_{b^*}^{a^*}\big]_{[n_r^*:~]}^{\rm H}\right)^2} ,
\end{align}
 where $\Re\{\cdot \}$ denotes the real part of a complex number, and
 $\widetilde{\bm{\Upsilon}}_{b^*,\emptyset n_r,n_r^*}^{a^*}$ is given by
\begin{align}\label{eq55}
 \widetilde{\bm{\Upsilon}}_{b^*,\emptyset n_r,n_r^*}^{a^*} =& \bm{\Upsilon}_{b^*,\emptyset n_r,n_r^*}^{a^*} 
  \left[\bm{\Theta}_{b^*}^{a^*}\right]_{(n_r,n_r)}\bm{\Upsilon}_{b^*,\emptyset n_r,n_r^*}^{a^*} .
\end{align}			
 Thus by recalling Lemmas~\ref{L3} and \ref{L4}, we have the following approximations
\begin{align}
 & \big[\widehat{\bm{H}}_{b^*}^{a^*}\big]_{[n_r^*:~]}\bm{\Upsilon}_{b^*,\emptyset n_r,n_r^*}^{a^*}
  \big[\widehat{\bm{H}}_{b^*}^{a^*}\big]_{[n_r^*:~]}^{\rm H} = \vartheta_{b^*,n_r^*}^{a^*} , \label{eq56} \\
 & \big[\bm{H}_{b^*}^{a^*}\big]_{[n_r^*:~]} \bm{\Upsilon}_{b^*,\emptyset n_r,n_r^*}^{a^*}
  \big[\widehat{\bm{H}}_{b^*}^{a^*}\big]_{[n_r^*:~]}^{\rm H} = \vartheta_{b^*,n_r^*}^{a^*} , \label{eq57} \\
 & \big[\bm{H}_{b^*}^{a^*}\big]_{[n_r^*:~]}\widetilde{\bm{\Upsilon}}_{b^*,\emptyset n_r,n_r^*}^{a^*}
  \big[\bm{H}_{b^*}^{a^*}\big]_{[n_r^*:~]}^{\rm H} = \text{Tr}\left\{\big[\bm{\Omega}_{b^*}^{a^*}\big]_{(n_r^*,n_r^*)}
  \widetilde{\bm{\Upsilon}}_{b^*,\emptyset n_r,n_r^*}^{a^*}\right\} , \label{eq58} \\
 & \big[\widehat{\bm{H}}_{b^*}^{a^*}\big]_{[n_r^*:~]} \widetilde{\bm{\Upsilon}}_{b^*,\emptyset n_r,n_r^*}^{a^*}
  \big[\widehat{\bm{H}}_{b^*}^{a^*}\big]_{[n_r^{*}:~]}^{\rm H}
  = \text{Tr}\left\{\big[\bm{\Theta}_{b^*}^{a^*}\big]_{(n_r^*,n_r^*)}
  \widetilde{\bm{\Upsilon}}_{b^*,\emptyset n_r,n_r^*}^{a^*}\right\} , \label{eq59} \\
 & \big[\widehat{\bm{H}}_{b^*}^{a^*}\big]_{[n_r^*:~]} \widetilde{\bm{\Upsilon}}_{b^*,\emptyset n_r,n_r^*}^{a^*}
  \big[\bm{H}_{b^*}^{a^*}\big]_{[n_r^*:~]}^{\rm H} =
  \text{Tr}\left\{\big[\bm{\Theta}_{b^*}^{a^*}\big]_{(n_r^*,n_r^*)}
  \widetilde{\bm{\Upsilon}}_{b^*,\emptyset n_r,n_r^*}^{a^*}\right\} , \label{eq60}
\end{align}
 where $\big[\bm{\Omega}_{b^*}^{a^*}\big]_{(n_r^*,n_r^*)}\in\mathbb{C}^{N_t\times N_t}$ is given by
\begin{align}\label{eq61}
 \big[\bm{\Omega}_{b^*}^{a^*}\big]_{(n_r^*,n_r^*)} =& \nu^2\big[\bm{M}_{b^*}^{a^*}\big]_{(n_r^*,n_r^*)}
  + \big[\bm{R}_{b^*}^{a^*}\big]_{(n_r^*,n_r^*)} \,\,. 
\end{align}	
 First substituting (\ref{eq56}) to (\ref{eq60}) into (\ref{eq54}) and then substituting
 the result into (\ref{eq49}), we obtain
\begin{align}\label{eq62}
 & \mathcal{E}\left\{\left| \big[\bm{H}_{b^*}^{a^*}\big]_{[n_r^*:~]}
  \big[\bm{V}_{b^*}^{a^*}\big]_{[~:n_r]}\right|^2\right\} =
  \frac{\text{Tr}\left\{\big[\bm{\Omega}_{b^*}^{a^*}\big]_{(n_r^*,n_r^*)}
  \widetilde{\bm{\Upsilon}}_{b^*,\emptyset n_r,n_r^*}^{a^*}\right\}}{(1 + \frac{1}{N_t}
  \vartheta_{b^*,n_r}^{a^*})^2} \nonumber \\
 & \hspace*{5mm}-\frac{2\! \Re\! \left\{\! \frac{1}{N_t}\vartheta_{b^*,n_r^*}^{a^*}
  \text{Tr}\left\{\big[\bm{\Theta}_{b^*}^{a^*}\big]_{(n_r^*,n_r^*)}
  \widetilde{\bm{\Upsilon}}_{b^*,\emptyset n_r,n_r^*}^{a^*}\right\}\! \right\}}
  {(1 + \frac{1}{N_t}\vartheta_{b^*,n_r^*}^{a^*})(1 + \frac{1}{N_t}\vartheta_{b^*,n_r}^{a^*})^2} 
  \! + \! \frac{\left(\! \frac{1}{N_t}\vartheta_{b^*,n_r^*}^{a^*}\! \right)^2 \!
  \text{Tr}\left\{\! \big[\bm{\Theta}_{b^*}^{a^*}\big]_{(n_r^*,n_r^*)} \!
  \widetilde{\bm{\Upsilon}}_{b^*,\emptyset n_r,n_r^*}^{a^*}\! \right\}}{\left(1 + \frac{1}{N_t}
  \vartheta_{b^*,n_r^*}^{a^*}\right)^2\big(1 + \frac{1}{N_t}\vartheta_{b^*,n_r}^{a^*}\big)^2} .
\end{align}

 Thirdly, according to Lemmas~\ref{L1} and \ref{L3}, we have
\begin{align}\label{eq63}
 \mathcal{E}\left\{\left| \big[\bm{H}_{b^*}^a\big]_{[n_r^*:~]}
  \big[\bm{V}_{b^a}^a\right]_{[~:n_r]}\big|^2\right\} =&
  \mathcal{E}\Bigg\{\Bigg|\frac{\big[\bm{H}_{b^*}^a\big]_{[n_r^*:~]}\bm{\Upsilon}_{b^a,\emptyset n_r}^a
  \big[\widehat{\bm{H}}_{b^a}^a\big]_{[n_r:~]}^{\rm H}}{1 + \frac{1}{N_t}
  \big[\widehat{\bm{H}}_{b^a}^a\big]_{[n_r:~]}\bm{\Upsilon}_{b^a,\emptyset n_r}^a
  \big[\widehat{\bm{H}}_{b^a}^a\big]_{[n_r:~]}^{\rm H}}\Bigg|^2\Bigg\} \nonumber \\
 =& \frac{\text{Tr}\left\{\bm{\Upsilon}_{b^a,\emptyset n_r}^a\left[\bm{\Theta}_{b^a}^a\right]_{(n_r,n_r)}
  \bm{\Upsilon}_{b^a,\emptyset n_r}^a\left[\bm{\Omega}_{b^*}^a\right]_{(n_r^*,n_r^*)}\right\}}
  {(1 + \frac{1}{N_t}\vartheta_{b^a,n_r}^{a})^2} , 
\end{align}	
 where $\vartheta_{b^a,n_r}^a$, $\left[\bm{\Theta}_{b^a}^a\right]_{(n_r,n_r)}\in\mathbb{C}^{N_t\times N_t}$
 and $\left[\bm{\Omega}_{b^*}^a\right]_{(n_r^*,n_r^*)}\in\mathbb{C}^{N_t\times N_t}$ are given 
 respectively by 
\begin{align}
 \vartheta_{b^a,n_r}^a =& \text{Tr}\left\{\big[\bm{\Theta}_{b^a}^a\big]_{(n_r,n_r)}
  \bm{\Upsilon}_{b^a,\emptyset n_r}^a\right\} , \label{eq64} \\
 \big[\bm{\Theta}_{b^a}^a\big]_{(n_r,n_r)} =& \nu^2\left[\bm{M}_{b^a}^a\right]_{(n_r,n_r)} +
  \big[\bm{\Phi}_{b^a}^a\big]_{(n_r,n_r)} , \label{eq65} \\
 \big[\bm{\Omega}_{b^*}^a\big]_{(n_r^*,n_r^*)} =& \nu^2\left[\bm{M}_{b^*}^a\right]_{(n_r^*,n_r^*)}
  + \big[\bm{R}_{b^*}^a\big]_{(n_r^*,n_r^*)}, \label{eq66}
\end{align}
 while $\bm{\Upsilon}_{b^a,\emptyset n_r}^a$ has a similar form as $\bm{\Upsilon}_{b^*,\emptyset n_r}^{a^*}$,
 which is given by
\begin{align}\label{eq67}
 \bm{\Upsilon}_{b^a,\emptyset n_r}^a\! =& \bigg(\xi_{b^a}^a\! \left(\! \bm{I}_{N_t}\! +\! \tilde{\delta}_{b^a}^a
  \varsigma^2 \bm{\Phi}_{a}\right)\! +\! \frac{1}{N_t}\nu^2\big(\bm{H}_{{\rm d},b^a}^a\big)^{\rm H}\!
  \left(\bm{I}_{N_r}\! +\! \delta_{b^a}^a \varsigma^2 \bm{\Phi}^{b^a}\right)^{-1}\! \bm{H}_{{\rm d},b^a}^a
  \! -\! \frac{1}{N_t}\big[\bm{\Theta}_{b^a}^a\big]_{(n_r,n_r)}\! \bigg)^{-1} \!\! .\!
\end{align}

 By substituting (\ref{eq49}), (\ref{eq62}) and (\ref{eq63}) into (\ref{eq16}), we obtain					
\begin{align}\label{eq68}
 & P_{{\rm I\&N}_{b^*,n_r^*}^{a^*}} = \sigma_w^2 + P_{r,b^*}^{a^*}\text{Tr}\Bigg\{
  \frac{\bm{\Upsilon}_{b^*,\emptyset n_r^*}^{a^*} \big[\bm{\Theta}_{b^*}^{a^*}\big]_{(n_r^*,n_r^*)}
  \bm{\Upsilon}_{b^*,\emptyset n_r^*}^{a^*} \big[\bm{\Xi}_{b^*}^{a^*}\big]_{(n_r^*,n_r^*)}}
  {\big(1 + \vartheta_{b^*,n_r^*}^{a^*}\big)^2} \Bigg\} \nonumber \\
 & \hspace*{5mm}+ P_{r,b^*}^{a^*}\sum\limits_{n_r\neq n_r^*} \Bigg(\frac{
  \text{Tr}\Big\{\big[\bm{\Omega}_{b^*}^{a^*}\big]_{(n_r^*,n_r^*)}
  \widetilde{\bm{\Upsilon}}_{b^*,\emptyset n_r,n_r^*}^{a^*}\Big\} }{ \big(1 + \frac{1}{N_t}
  \vartheta_{b^*,n_r}^{a^*}\big)^2 } -
  \frac{ 2\Re\Big\{\frac{1}{N_t}\vartheta_{b^*,n_r^*}^{a^*}\text{Tr}\big\{
  \big[\bm{\Theta}_{b^*}^{a^*}\big]_{(n_r^*,n_r^*)}\widetilde{\bm{\Upsilon}}_{b^*,\emptyset n_r,n_r^*}^{a^*}\big\}
  \Big\}}{\big(1 + \frac{1}{N_t}\vartheta_{b^*,n_r^*}^{a^*}\big)\big(1 + \frac{1}{N_t}
  \vartheta_{b^*,n_r}^{a^*}\big)^2} \nonumber \\
 & \hspace*{5mm}+ \frac{\Big(\frac{1}{N_t}\vartheta_{b^*,n_r^*}^{a^*}\Big)^2
  \text{Tr}\Big\{\big[\bm{\Theta}_{b^*}^{a^*}\big]_{(n_r^*,n_r^*)}
  \widetilde{\bm{\Upsilon}}_{b^*,\emptyset n_r,n_r^*}^{a^*}\Big\} }{\Big(1 + \frac{1}{N_t}\vartheta_{b^*,n_r^*}^{a^*}\Big)^2
  \big(1 + \frac{1}{N_t}\vartheta_{b^*,n_r}^{a^*})^2}\Bigg) \nonumber \\
 & \hspace*{5mm}+ \sum\limits_{a=1}^A P_{r,b^*}^a \sum\limits_{n_r=1}^{N_r} \frac{\text{Tr}\Big\{
  \bm{\Upsilon}_{b^a,\emptyset n_r}^a \big[\bm{\Theta}_{b^a}^a\big]_{(n_r,n_r)}
  \bm{\Upsilon}_{b^a,\emptyset n_r}^a \big[\bm{\Omega}_{b^{*}}^{a}\big]_{(n_r^*,n_r^*)}
  \Big\}}{\big(1 + \frac{1}{N_t}\vartheta_{b^a,n_r}^a\big)^2} .
\end{align}

\subsection{Achievable Rate and Optimal Regularization Parameter}\label{S3.4}
		
 Finally, upon substituting (\ref{eq46}) as well as (\ref{eq68}) into (\ref{eq17}) and then
 using the result in (\ref{eq18}), we arrive at the closed-form achievable transmission rate
 per antenna $C_{b^*}^{a^*}$, which is our performance metric for designing the distance
 thresholds for the distance-based  ACM scheme.

 Since it is intractable to obtain an analytic optimal regularization parameter that maximizes 
 the achievable transmission rate per antenna, we consider the alternative mean-square-error
 for detecting the transmitted data vector $\bm{X}^{a^*}$ by aircraft $b^*$, which is
 given by
\begin{align}\label{eq69}
 & \mathcal{J}\big(\xi_{b^*}^{a^*}\big) = \mathcal{E}\Bigg\{\Bigg\|\frac{1}{N_t}\bm{H}_{b^*}^{a^*}
  \bm{V}_{b^*}^{a^*} \bm{X}^{a^*} + \frac{1}{N_t}\sum\limits_{a=1}^A \sqrt{\frac{P_{r,b^*}^a}{P_{r,b^*}^{a^*}}}
  \bm{H}_{b^*}^a\bm{V}_{b^a}^a \bm{X}^a + \frac{1}{N_t\sqrt{P_{r,b^*}^{a^*}}} \bm{W}_{b^*} -
  \bm{X}^{a^*}\Bigg\|^2\Bigg\} \nonumber \\
 & \hspace*{5mm}= \mathcal{E} \Bigg\{\! \bigg\|\frac{1}{N_t}\bm{H}_{b^*}^{a^*}\bm{V}_{b^*}^{a^*}
  \bm{X}^{a^*} \!\! -\! \bm{X}^{a^*}\bigg\|^2\!\! +\! \sum\limits_{a=1}^A \left\|\frac{1}{N_t}
  \sqrt{\frac{P_{r,b^*}^a}{P_{r,b^*}^{a^*}}}\bm{H}_{b^*}^a\bm{V}_{b^a}^a \bm{X}^a\right\|^2 \!\! + \!
  \Bigg\|\frac{1}{N_{t}\sqrt{P_{r,b^*}^{a^*}}}\bm{W}_{b^*}\Bigg\|^2 \Bigg\} \! .\!
\end{align}
 As detailed in Appendix~\ref{Apb}, the  closed-form optimal regularization parameter that
 minimizes the  mean-square data detection error (\ref{eq69}) is given by
\begin{align}\label{eq70}
 \big(\xi_{b^{*}}^{a^{*}}\big)^{\star} =& \frac{1}{N_t} \sum\limits_{n_t=1}^{N_t} \widetilde{\varphi}_{n_t},
\end{align}
 with
\begin{align}\label{eq71}
 \widetilde{\varphi}_{n_t} =& \sum\limits_{k=1}^{N_r} \bm{\Xi}_{b^*}^{a^*}\lvert_{[n_t+(k-1)N_t,n_t+(k-1)N_t]} ,
\end{align}
 in which $\bm{\Xi}_{b^*}^{a^*}\lvert_{[i,i]}$ denotes the $i$-th row and $i$-th column
 element of $\bm{\Xi}_{b^*}^{a^*}$.

\section{RZF-TPC Aided and Distance-Based ACM Scheme}\label{S4}

 The framework of the RZF-TPC aided and distance-based ACM scheme, which is depicted
 in the middle of Fig.~\ref{FIG1}, is similar to that of the EB-TPC aided and
 distance-based ACM scheme given in \cite{zhang2017adaptive}. The main difference is
 that here we adopt the much more powerful TPC solution at the transmitter. Given the
 system parameters, including the total system bandwidth $B_{\rm total}$, the number
 of subcarriers $N$, the number of CP samples $N_{\rm cp}$, the set of modulation
 constellations and the set of channel codes, the number of ACM modes $K$ together with
 the set of switching thresholds $\{d_k\}_{k=0}^K$ can now be designed, where
 $d_0=D_{\max}$ and $D_{\max}$ is the maximum communication distance, while $d_K=D_{\min}$
 and $D_{\min}$ is the minimum safe separation distance of aircraft. The online operations
 of the RZF-TPC aided and distance-based ACM transmission can then be summarized below.

\begin{itemize}
\item[1)] In the pilot training phase, aircraft $a^*$ estimates the channel matrix
 $\bm{H}^{b^*}_{a^*}$ between aircraft $b^*$ and aircraft $a^*$ based on the pilots sent
 by $b^*$.
\item [2)] Aircraft $a^*$ acquires the channel $\bm{H}_{b^*}^{a^*}$ for data transmission
 by exploiting the channel reciprocity, and generates the RZF-TPC matrix $\bm{V}_{b^*}^{a^*}$
 according to (\ref{eq11}) and (\ref{eq12}).
\item[3)] Based on the distance $d_{b^*}^{a^*}$ between aircraft $a^*$ and $b^*$
 measured by its DME, aircraft $a^*$ selects an ACM mode for data transmission according to
\begin{align}\label{eq72}
 & \text{If } d_k\le d_{b^*}^{a^*} < d_{k-1}: \text{ choose mode } k , \, k \in \{1,2,\cdots , K\} .
\end{align}
\end{itemize}
 Note that the scenarios of $d_{b^*}^{a^*}\ge D_{\max}$ and $d_{b^*}^{a^*}\le D_{\min}$ are
 not considered, since there is no available communication link, when two aircraft are beyond
 the maximum communication range, while the minimum flight-safety separation must be maintained.

\begin{table*}[bp!]
\vspace*{-6mm}
\caption{A design example of RZF-TPC aided and distance-based ACM with $N_t=32$ and $N_r=4$.
 The other system parameters for this ACM are listed in Table~\ref{Tab3}.}
\vspace*{-8mm}
\begin{center}
\begin{tabular}{|C{1.2cm}|C{1.3cm}|C{1.3cm}|C{2.3cm}|C{2.3cm}|C{2.3cm}|C{2.3cm}|}
\hline
 Mode $k$ & Modulation & Code rate & Spectral efficiency (bps/Hz) & Switching threshold $d_k$ (km)
   & Data rate per receive antenna (Mbps) & Total data rate (Mbps) \\ \hline
 1 & BPSK   & 0.488 & 0.459 & 550  & 2.756  & 11.023 \\ \hline
 2 & QPSK   & 0.533 & 1.003 & 450  & 6.020  & 24.079 \\ \hline
 3 & QPSK   & 0.706 & 1.329 & 300  & 7.934  & 31.895 \\ \hline
 4 & 8-QAM  & 0.635 & 1.793 & 210  & 10.758 & 43.031 \\ \hline
 5 & 8-QAM  & 0.780 & 2.202 & 130  & 13.214 & 52.857 \\ \hline
 6 & 16-QAM & 0.731 & 2.752 & 90   & 16.512 & 66.048 \\ \hline
 7 & 16-QAM & 0.853 & 3.211 & 35   & 19.258 & 77.071 \\ \hline
 8 & 32-QAM & 0.879 & 4.137 & 5.56 & 24.819 & 99.275 \\ \hline
\end{tabular}
\end{center}
\label{Tab1}
\vspace*{-2mm}
\end{table*}

\begin{table*}[tp!]
\caption{A design example of RZF-TPC aided and distance-based ACM with $N_t=64$ and $N_r=4$.
 The other system parameters for this ACM are listed in Table~\ref{Tab3}.}
\vspace*{-8mm}
\begin{center}
\begin{tabular}{|C{1.2cm}|C{1.3cm}|C{1.3cm}|C{2.3cm}|C{2.3cm}|C{2.3cm}|C{2.3cm}|}
\hline
 Mode $k$ & Modulation & Code rate & Spectral efficiency (bps/Hz) & Switching threshold $d_k$ (km)
   & Data rate per receive antenna (Mbps) & Total data rate (Mbps) \\ \hline
 1 & QPSK   & 0.706 & 1.323 & 500  & 7.974  & 31.895  \\ \hline
 2 & 8-QAM  & 0.642 & 1.813 & 400  & 10.876 & 43.505  \\ \hline
 3 & 8-QAM  & 0.780 & 2.202 & 300  & 13.214 & 52.857  \\ \hline
 4 & 16-QAM & 0.708 & 2.665 & 190  & 15.993 & 63.970   \\ \hline
 5 & 16-QAM & 0.853 & 3.211 & 90   & 19.268 & 77.071  \\ \hline
 6 & 32-QAM & 0.831 & 3.911 & 35   & 23.464 & 93.854  \\ \hline
 7 & 64-QAM & 0.879 & 4.964 & 5.56 & 29.783 & 119.130 \\ \hline
\end{tabular}
\end{center}
\label{Tab2}
\vspace*{-10mm}
\end{table*}

\begin{figure*}[bp!]
\vspace*{-7mm}
\begin{center}
\subfigure[$N_t = 32, N_r = 4$]{\includegraphics[width=0.47\textwidth]{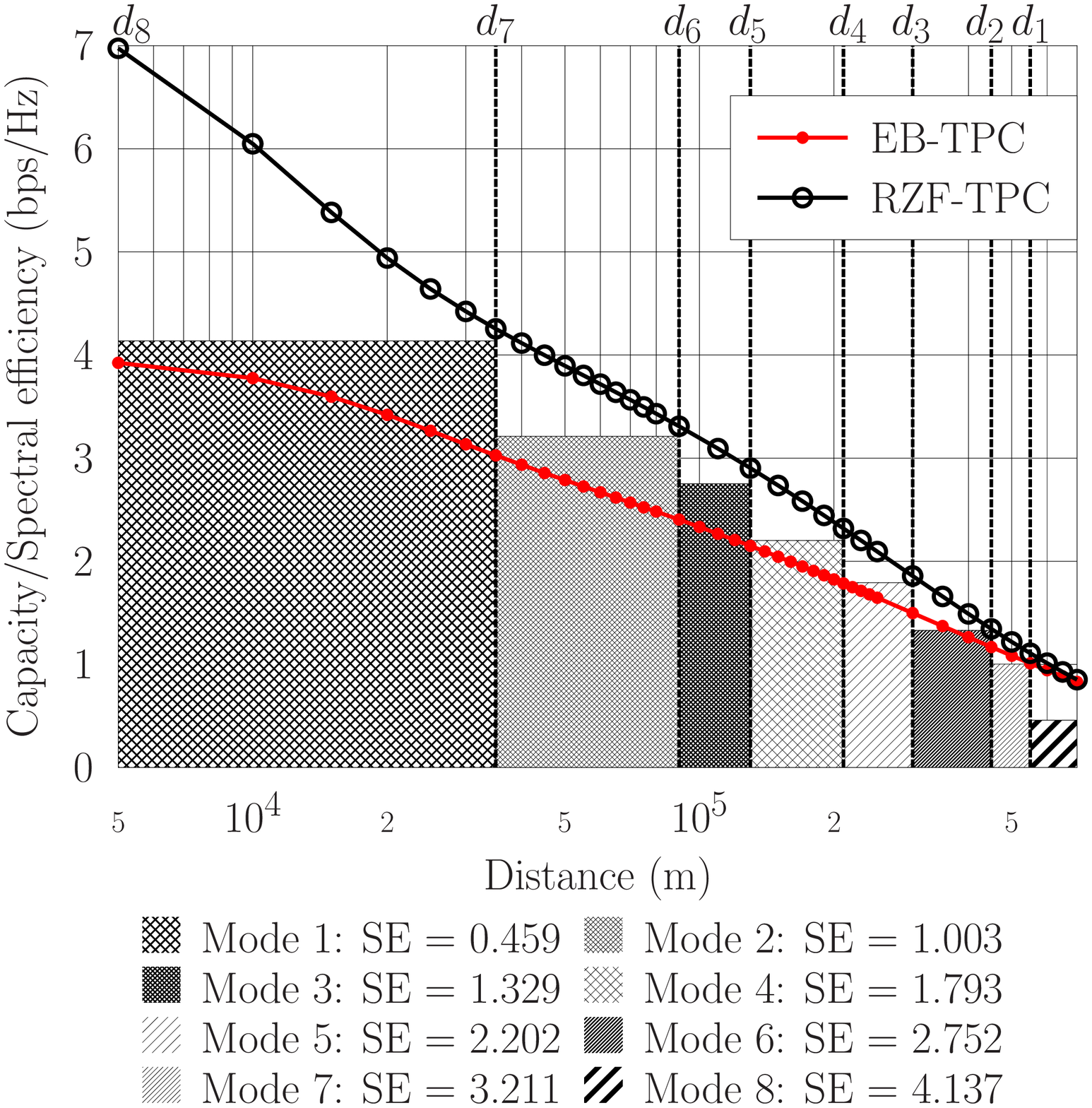} 
\label{FIG2a}}
\subfigure[$N_t = 64, N_r = 4$]{\includegraphics[width=0.47\textwidth]{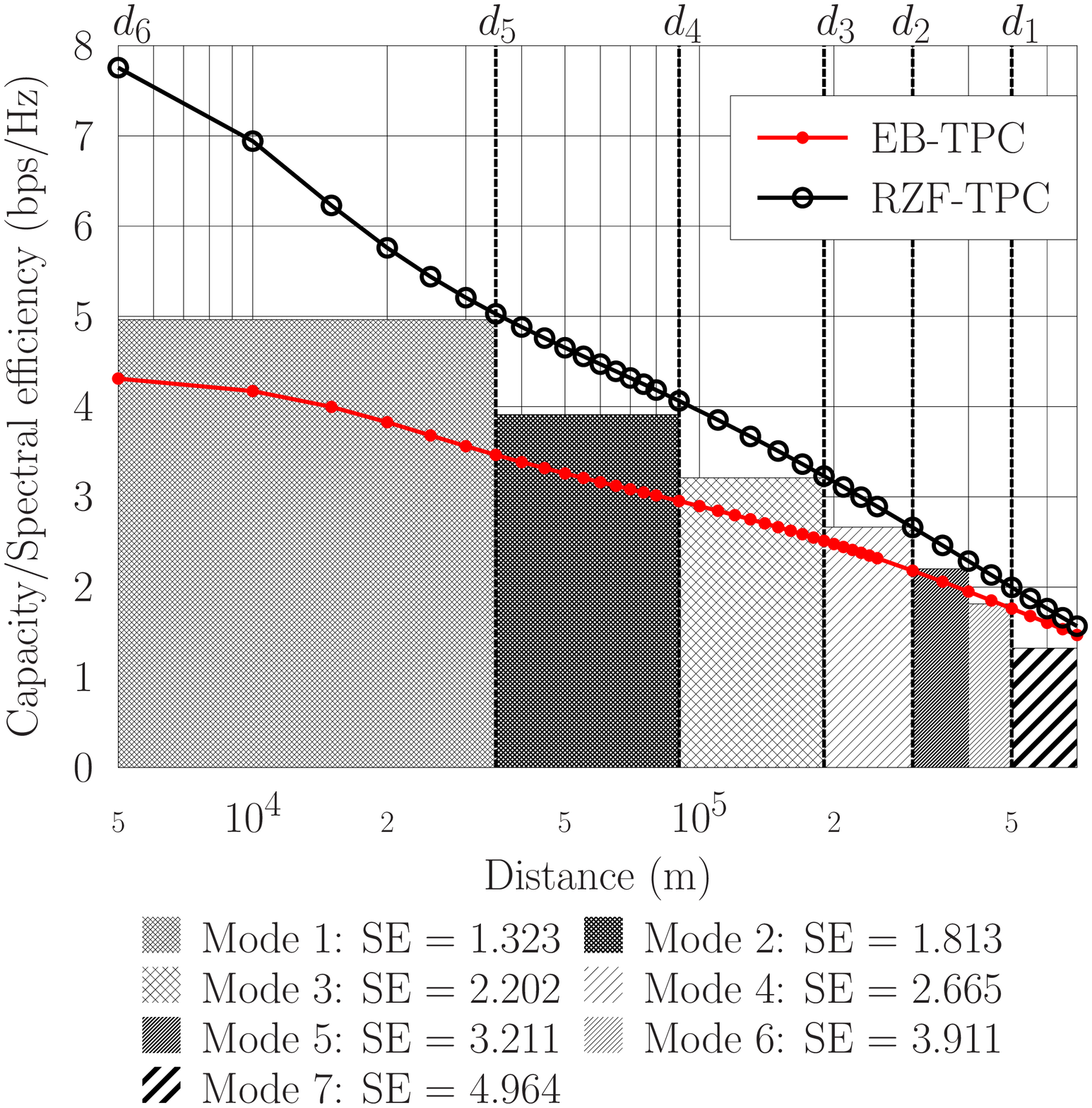}
\label{FIG2b}}
\end{center}
\vspace{-6mm}
\caption{Two examples of designing the RZF-TPC aided and distance-based ACM scheme. The
 switching thresholds and achievable system throughputs for (a) and (b) are listed in 
 Tables~\ref{Tab1} and \ref{Tab2}, respectively.}
\label{FIG2}
\vspace{-2mm}
\end{figure*}

 Tables~\ref{Tab1} and \ref{Tab2} provide two design examples of the RZF-TPC aided
 and  distance-based ACM in conjunction with $(N_t,N_r)=(32,4)$ and $(N_t,N_r)=(64,4)$,
 respectively. The RZF-TPC aided and distance-based ACM consists of $K$ ACM modes for
 providing $K$ data rates. The modulation schemes and code rates are selected from the
 second generation VersaFEC \cite{VersaFEC}, which is well designed to provide high
 performance and low latency ACM. The SE of mode $k$, ${\rm SE}_k$, is given by
\begin{align}\label{eq73}
 {\rm SE}_k =& r_c\log_2(M)\frac{N}{N + N_{\rm cp}} ,
\end{align}
 where $M$ is the modulation order, $r_c$ is the coding rate and the data rate $r_{{\rm pDRA},k}$ per
 DRA of mode $k$, is given by
\begin{align}\label{eq74}
 r_{{\rm pDRA},k} =& B_{\rm total} \cdot {\rm SE}_k ,
\end{align}
 while the total data rate $r_{{\rm total},k}$ of mode $k$, is given by
\begin{align}\label{eq75}
 r_{{\rm total},k} =& N_r \cdot r_{{\rm pDRA},k}  .
\end{align}
 More explicitly, Fig.~\ref{FIG2a} illustrates how the $K=8$ ACM modes are designed for
 the example of Table~\ref{Tab1}. Explicitly, the 8 switching thresholds of Table~\ref{Tab1}
 are determined so that the SEs of the corresponding ACM modes is just below the SE curve
 of the RZF-TPC. The $K=7$ ACM modes of Table~\ref{Tab2} are similarly determined, as 
 illustrated in Fig.~\ref{FIG2b}. Fig.~\ref{FIG2} also confirms that the RZF-TPC aided and
 distance-based ACM scheme significantly outperforms the EB-TPC aided and distance-based
 ACM scheme of \cite{zhang2017adaptive}.

\begin{figure}[bp!]
\vspace*{-8mm}
\begin{center}
 \includegraphics[width=0.5\textwidth]{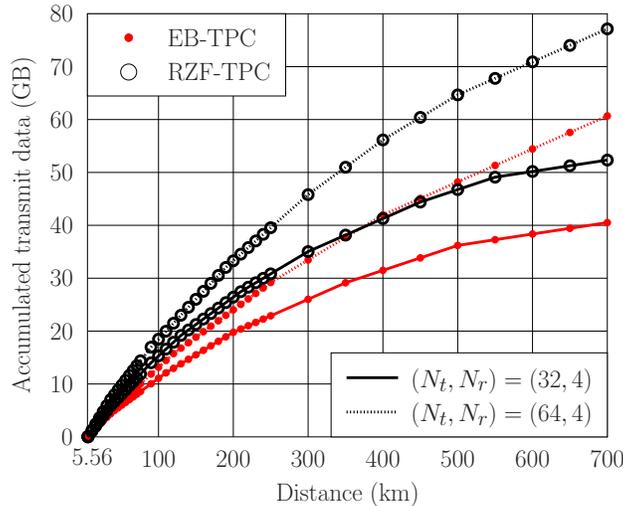} 
 \end{center}
\vspace{-12mm}
\caption{The accumulated transmitted data volume as the function of the distance
 $d_{b^*}^{a^*}$ between communicating aircraft $a^{*}$ and $b^{*}$. The  distances
 between the interfering aircraft and the desired receiving aircraft are uniformly
 distributed within the range of $\big[d_{b^*}^{a^*},~ D_{\max}\big]$, and other
 system parameters are specified in Table~\ref{Tab3}.}
\label{FIG3}
\vspace*{-3mm}
\end{figure}

 To further compare the RZF-TPC aided and distance-based ACM to the EB-TPC aided and
 distance-based ACM of \cite{zhang2017adaptive}, we quantitatively calculate their
 accumulated transmitted data volume $C_{\rm acc}$ that can be exchanged between
 aircraft $a^*$ and aircraft $b^*$ both flying at a typical cruising speed of $v=920$\,km/h
 in the opposite direction from the minimum separation distance $d_{\min}=5.56$\,km to
 the maximum communication distance $d_{\max}=740$\,km. Let the current distance
 $d^{a^*}_{b^*}$ between aircraft $a^*$ and aircraft $b^*$ be $d^{a^*}_{b^*}\in
 \big(d_{K-(k^*-1)}, ~ d_{K-k^*}\big]$. The accumulated data volume transmitted from
 $a^*$ to $b^*$ for $d_{K-(k^*-1)} < d^{a^*}_{b^*} \le d_{K-k^*}$ can be calculated by
\begin{align}\label{eq76}
 C_{\rm acc}\big(d^{a^*}_{b^*}\big) =& \frac{ d^{a^*}_{b^*} - d_{K- k^*}}{2v}
  r_{{\rm total},(K-k^*)} + \sum\limits_{k=1}^{k^*}\frac{d_{K-k} - d_{K-(k-1)}}{2v}
  r_{{\rm total},(K-(k-1))} .
\end{align}
 The accumulated transmitted data volumes expressed in gigabyte (GB) of the RZF-TPC
 aided and distance-based ACM are compared with the EB-TPC aided and distance-based ACM
 in Fig.~\ref{FIG3}, for $(N_t,N_r)=(32,4)$ and $(N_t,N_r)=(64,4)$. As expected, the
 achievable accumulated transmitted data volume of the RZF-TPC aided and distance-based
 ACM is significantly higher than that of the EB-TPC aided and distance-based ACM. In
 particular, when aircraft $a^*$ and $b^*$ fly over the communication distance, from
 $d_{b^*}^{a^*}=5.56$\,km to $d_{b^*}^{a^*}=740$\,km  taking a period of about 24
 minutes, the RZF-TPC aided and distance-based ACM associated with $(N_t,N_r)=(64,4)$
 is capable of transmitting a total of about 77\,GB of data, while the EB-TPC aided and
 distance-based ACM with $(N_t,N_r) = (64,4)$ is only capable of transmitting about
 60\,GB of data. Note that (\ref{eq76}) can be revised to include any other scenario,
 by introducing the angle of bearing between two aircraft and their heading direction.

\begin{table}[bp!]
\vspace*{-7mm}
\caption{Default system parameters of the simulated AANET.}
\vspace*{-9mm}
\begin{center}
\begin{tabular}{l|l}
\hline\hline
 Number of interference aircraft $A$ & 4 \\ \hline
 Number of DRAs $N_r$ & 4 \\ \hline
 Number of DTAs $N_t$ & 32 \\ \hline
 Transmit power per antennas $P_t$ & 1 watt \\ \hline
 Number of total subcarriers $N$ & 512 \\ \hline
 Number of CPs $N_{\rm cp}$ & 32 \\ \hline
 Rician factor $K_{\rm Rice}$ & 5 \\ \hline
 System bandwidth $B_{\rm total}$ & 6 MHz \\ \hline
 Carrier frequency & 5 GHz \\ \hline
 Correlation factor between antennas $\rho$ & 0.1 \\ \hline
 Noise figure at receiver $F$ & 4 dB \\ \hline
 Distance $d_{b^*}^{a^*}$ between communicating aircraft $a^*$ and $b^*$ & 10 km \\ \hline
 Minimum communication distance $D_{\min}$  & 5.56 km \\ \hline
 Maximum communication distance $D_{\max}$  & 740 km \\ \hline\hline
\end{tabular}
\end{center}
\label{Tab3}
\vspace*{-3mm}
\end{table}

\section{Simulation Study}\label{S5}

 To further evaluate the achievable performance of the proposed RZF-TPC aided
 and distance-based ACM scheme as well as to investigate the impact of the key
 system parameters, we consider an AANET consisting of $(A+2)$ aircraft, with
 two desired communicating aircraft and $A$ interfering aircraft. Each aircraft
 is equipped with $N_t$ DTAs and $N_r$ DRAs. The network is allocated
 $B_{\rm total}=6$\,MHz bandwidth at the carrier frequency of 5\,GHz. This
 bandwidth is reused by every aircraft and it is divided into $N=512$ subcarriers.
 The CP samples are $N_{\rm cp}=32$. The transmit power per antenna is
 $P_t=1$\,Watt. The default system parameters are summarized in Table~\ref{Tab3}.
 Unless otherwise specified, these default parameters are used. The deterministic
 part of the Rician channel is generated according to the model given in
 \cite{jin2010low}, which satisfies $\text{Tr}\left\{\bm{H}_{{\rm d},b}^a\big(
 \bm{H}_{{\rm d},b}^a\big)^{\rm H}\right\}=N_t N_r$. The scattering component of
 the Rician channel $\bm{H}_{\text{r},b}^a\in \mathbb{C}^{N_r\times N_t}$ is
 generated according to (\ref{eq3}). As mentioned previously, the DRAs are
 uncorrelated and, therefore, we have $\bm{R}_b=\bm{I}_{N_r}$. The spatial
 correlation matrix of the DTAs is generated according to
\begin{align}\label{eq78}
 \bm{R}^a\lvert_{[m,n]} =& \big(\bm{R}^a\lvert_{[n,m]}\big)^{\ddag} =(t\rho)^{|m-n|} , \, 1\le n,m\le N_t ,
\end{align}
 where $( \cdot )^{\ddag}$ denotes the conjugate operator, $t$ is a complex number
 with $|t|=1$ and $\rho$ is the magnitude of the correlation coefficient that is
 determined by the antenna element spacing \cite{lee2015antenna}.

 In the following investigation of the achievable throughput by the RZF-TPC aided
 and distance-based ACM scheme, `Theoretical results' are the throughputs calculated
 using (\ref{eq18}) using the perfect knowledge of $\left[\bm{M}_{b^a}^a\right]_{(n_r,n_r)}$
 and $\left[\bm{M}_{b^*}^a\right]_{(n_r^*,n_r^*)}$ in (\ref{eq68}), and the
 `Approximate results' are the throughputs calculated using (\ref{eq18}) with both
 $\left[\bm{M}_{b^a}^a \right]_{(n_r, n_r)}$ and $\left[\bm{M}_{b^*}^a\right]_{(n_r^*,n_r^*)}$
 substituted by $\left[\bm{M}_{b^*}^{a^*}\right]_{(n_r^*,n_r^*)}$ in (\ref{eq68}),
 while the `Simulation results' represent the Monte-Carlo simulation results. For the
 EB-TPC aided and distance-based ACM scheme, the `Theoretical results', `Approximate
 results' and `Simulation results' are defined similarly.

\begin{figure}[btp!]
\vspace*{-4mm}
\begin{center}
 \includegraphics[width=0.5\textwidth]{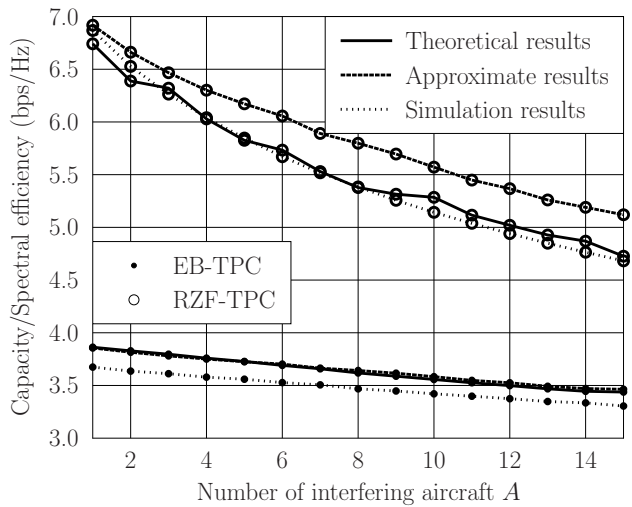} 
\end{center}
\vspace*{-11mm}
\caption{The achievable throughput per DRA as the function of the number of interfering
 aircraft $A$. The distances between the interfering aircraft and the desired receiving
 aircraft are uniformly distributed within the range of $\big[d_{b^*}^{a^*}, ~ D_{\max}\big]$,
 and the rest of the parameters are specified in Table~\ref{Tab3}.}
 \label{FIG4}
\vspace*{-4mm}
\end{figure}

 In Fig.~\ref{FIG4}, we investigate the achievable throughput per DRA as the function
 of the number of interfering aircraft $A$. Observe from Fig.~\ref{FIG4} that for the 
 RZF-TPC aided and distance-based ACM, the `Theoretical results' are closely matched by 
 the `Simulation results', which indicates that our theoretical analysis presented in
 Section~\ref{S3} is accurate. Furthermore, there is about 0.4 bps/Hz gap between the
 `Theoretical results' and the `Approximate results'. As expected, the achievable
 throughput degrades as the number of interfering aircraft increases. Moreover, the
 RZF-TPC aided and distance-based ACM scheme is capable of achieving significantly
 higher SE than the EB-TPC aided and distance-based ACM scheme. 

\begin{figure}[tbp!]
\vspace*{-4mm}
\begin{center}
\includegraphics[width=0.5\textwidth]{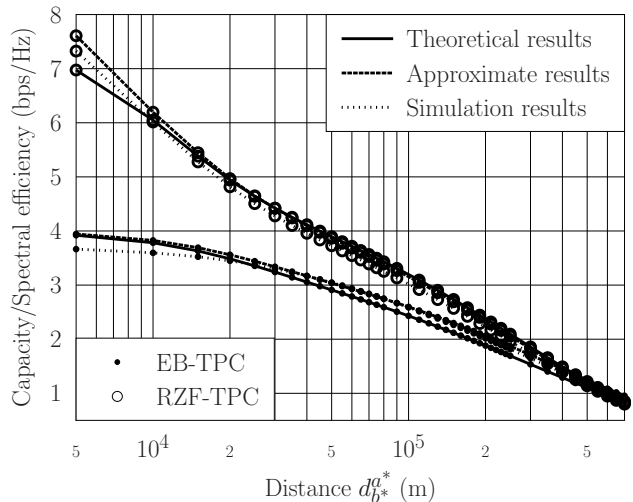} 
\end{center}
\vspace{-11mm}
\caption{The achievable throughput per DRA as a function of the distance $d_{b^*}^{a^*}$
 between the communicating aircraft $a^*$ and $b^*$. The distances between the interfering
 aircraft and the desired receiving aircraft are uniformly distributed within the range of
 $\big[d_{b^*}^{a^*},~ D_{\max}\big]$, and the rest of the parameters are specified in
 Table~\ref{Tab3}.}
 \label{FIG5}
\vspace{-4mm}
\end{figure}

 Fig.~\ref{FIG5} portrays the achievable throughput per DRA as the function of the distance
 $d_{b^*}^{a^*}$ between the communicating aircraft $a^*$ and $b^*$. Compared to the EB-TPC
 aided and distance-based ACM, the RZF-TPC aided and distance-based ACM is capable of
 achieving significantly higher SE, particularly at shorter distances. At the minimum
 distance of $d_{b^*}^{a^*}=5.56$\,km, the SE improvement is about 3\,bps/Hz, but the
 SE improvement becomes lower as the distance becomes longer. When the distance approaches
 the maximum communication range of 740\,km, both the schemes have a similar SE.

\begin{figure}[tbp!]
\vspace{-2mm}
\begin{center}
\includegraphics[width=0.5\textwidth]{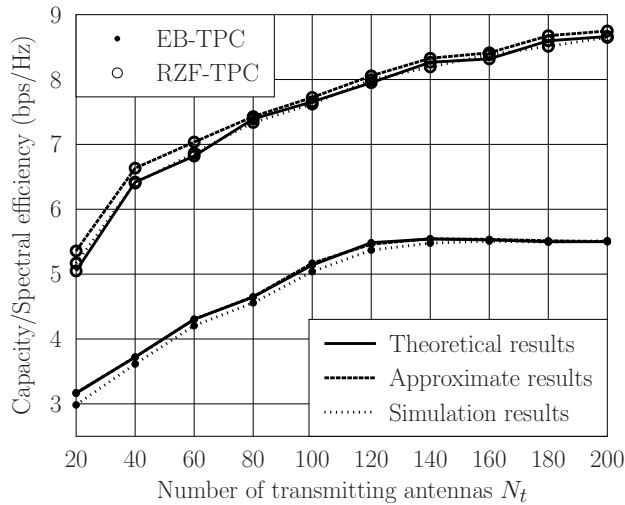}  
\end{center}
\vspace{-11mm}
\caption{The achievable throughput per DRA as the function of the number of DTAs $N_t$.
 The distances between the interfering aircraft and the desired receiving aircraft are
 uniformly distributed within the range of $\big[d_{b^*}^{a^*},~ D_{\max}\big]$, and
 the rest of the parameters are specified in Table~\ref{Tab3}.}
 \label{FIG6}
\vspace{-8mm}
\end{figure}

 Fig.~\ref{FIG6} shows the impact of the number of DTAs $N_t$ on the achievable throughput.
 As expected, the achievable throughput increases upon increasing $N_t$. Observe from
 Fig.~\ref{FIG6} that for the RZF-TPC aided and distance-based ACM, the `Theoretical
 results' are closely matched by the `Simulation results', while the `Theoretical results'
 are closely matched by the `Approximate results', when $N_t\ge 80$, but there exists a
 small gap between the `Theoretical results' and the `Approximate result' for $N_t < 80$.
 It can also be seen that the RZF-TPC aided and distance-based ACM achieves approximately
 3.0\,bps/Hz SE improvement over the EB-TPC aided and distance-based ACM.

\begin{figure}[tbp!]
\vspace{-3mm}
\begin{center}
\includegraphics[width=0.5\textwidth]{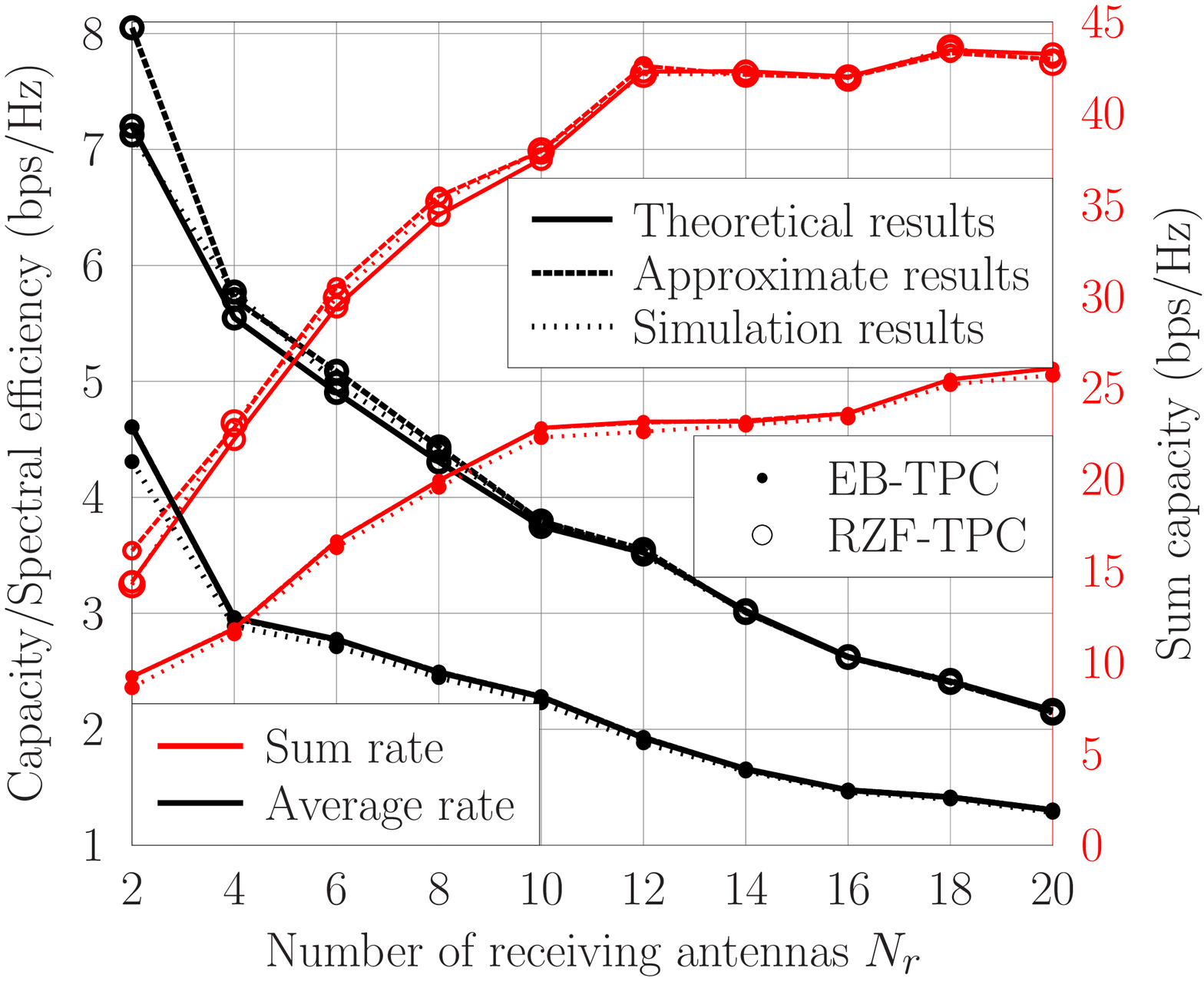} 
\end{center}
\vspace{-8mm}
\caption{The achievable throughput as the function of the number of DRAs $N_r$. The
 distances between the interfering aircraft and the desired receiving aircraft are
 uniformly distributed within the range of $\big[d_{b^*}^{a^*}, ~ D_{\text{max}}\big]$,
 and the rest of the parameters are specified in Table~\ref{Tab3}.}
\label{FIG7}
\vspace{-3mm}
\end{figure}

 The impact of the number of DRAs $N_r$ on the achievable throughput is studied in
 Fig.~\ref{FIG7}, where both the achievable throughput per antenna and the achievable
 sum rate of all the $N_r$ DRAs are plotted. Observe that the achievable throughput
 per antenna degrades upon increasing $N_r$, owing to the increase of the
 interference amongst the receive antennas. On the other hand, the achievable sum rate
 increases with $N_r$ due to the multiplexing gain. But the sum rate becomes 
 saturated for $N_r > 12$, because the increase in multiplexing gain is roughly cancelled
 by the increase of inter-antenna interference. Not surprisingly, the RZF-TPC aided and
 distance-based ACM significantly outperforms the EB-TPC aided and distance-based ACM.

\begin{figure}[htp!]
\vspace{-3mm}
\begin{center}
\includegraphics[width=0.5\textwidth]{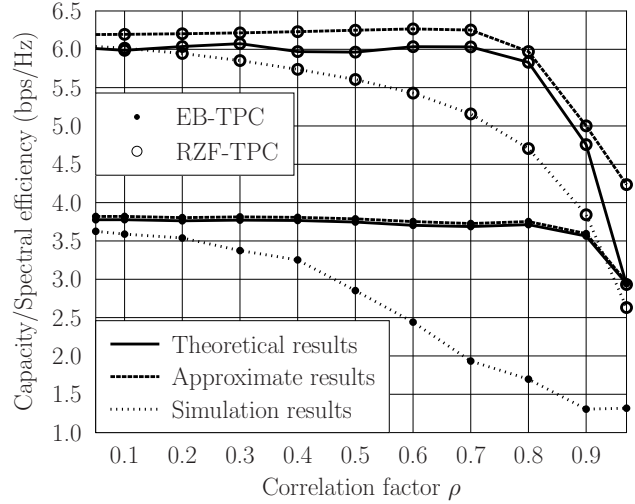} 
\end{center}
\vspace{-11mm}
\caption{The achievable throughput per DRA as the function of the correlation factor of
 DTAs $\rho$. The distances between the interfering aircraft and the desired receiving
 aircraft are uniformly distributed within the range of $\big[d_{b^*}^{a^*}, ~
 D_{\max}\big]$, and the rest of the parameters are specified in Table~\ref{Tab3}.}
\label{FIG8}
\vspace{-3mm}
\end{figure}

\begin{figure}[tbp!]
\vspace*{-1mm}
\begin{center}
\includegraphics[width=0.5\textwidth]{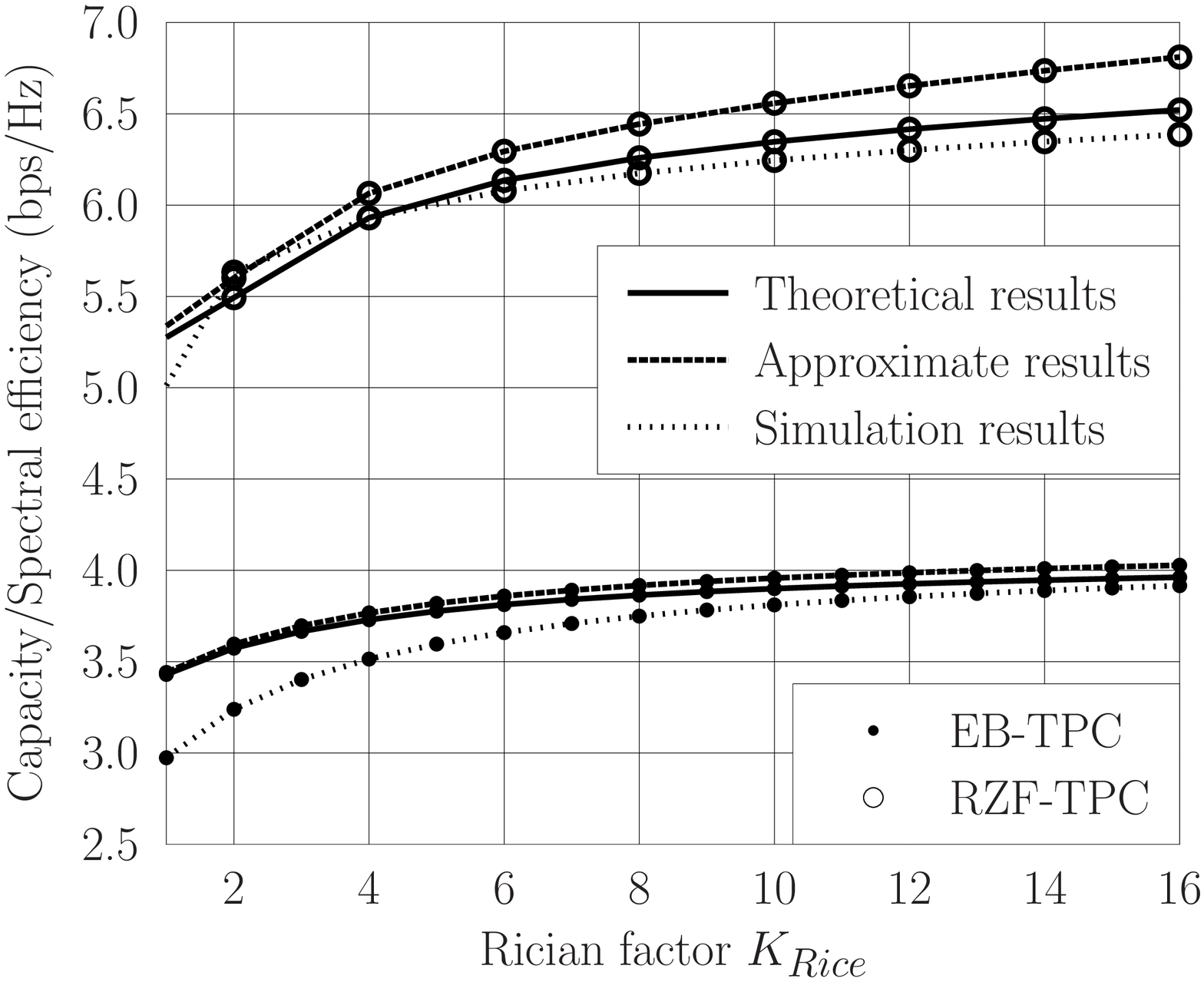} 
\end{center}
\vspace{-11mm}
\caption{The achievable throughput per DRA as the function of the Rician factor
 $K_{\rm Rice}$. The distances between the interfering aircraft and the desired
 receiving aircraft are uniformly distributed within the range of $\big[d_{b^*}^{a^*},
 ~ D_{\max}\big]$, and the rest of the parameters are specified in Table~\ref{Tab3}.}
\label{FIG9}
\vspace{-5mm}
\end{figure}

 The effect of the correlation factor $\rho$ of the DTAs on the achievable throughput
 per DRA is shown in Fig.~\ref{FIG8}. It can be observed that a higher correlation
 between DTAs results in lower achievable throughput. For the RZF-TPC aided and
 distance-based ACM,  the simulated throughput and the theoretical throughput are close
 for $\rho \le 0.4$, but there is a clear performance gap between the `Theoretical
 results' and the `Simulation results' for $\rho\ge 0.5$. For the EB-TPC aided and
 distance-based ACM, this performance gap between the `Theoretical results' and the
 `Simulation results' is even bigger and  it exists clearly over the range of
 $\rho\ge 0.3$. This indicates that for a higher correlation factor $\rho$, the simulated 
 SINR, which is the average over a number of realizations, may deviate considerably
 from  the theoretical SINR, which is the ensemble average. From Fig.~\ref{FIG8},
 it is clear that in addition to achieving a significantly better SE performance,
 the RZF-TPC aided and distance-based ACM can better deal with the problem caused by
 strong correlation among the DTAs than the EB-TPC aided and distance-based ACM.

 Fig.~\ref{FIG9} portrays the impact of the Rician factor $K_{\rm Rice}$ on the
 achievable throughput per DRA. It can be seen from Fig.~\ref{FIG9} that the achievable
 throughput per DRA increases upon increasing the Rician factor $K_{\rm Rice}$.
 Furthermore, the SE improvement of the RZF-TPC aided and distance-based ACM over the 
 EB-TPC aided and distance-based ACM also increases with $K_{\rm Rice}$. Specifically,
 the SE enhancement is about 1.9\,bps/Hz at $K_{\rm Rice}=2.0$ and this is increased
 to about 2.6\,bps/Hz for $K_{\rm Rice}=16.0$.

\section{Conclusions}\label{S6}

 A RZF-TCP aided and distance-based ACM scheme has been proposed for large-scale
 antenna array assisted aeronautical communications. For the design of powerful
 RZF-TCP, our theoretical contribution has been twofold. For the  first time, we have
 derived the analytical closed-form achievable data rate in the presence of both
 realistic channel estimation error and co-channel interference. Moreover, we have 
 explicitly derived the optimal regularization parameter that minimizes the
 mean-square detection error. With the aid of this closed-form data rate metric, we
 have designed a practical distance-based ACM scheme that switches its coding and
 modulation mode according to the distance between the communicating aircraft. Our
 extensive simulation study has quantified the impact of the key system parameters
 on the achievable throughput of the proposed RZF-TCP aided and distance-based ACM
 scheme. Our simulation results have confirmed the accuracy of our analytical results.
 Moreover, both our theoretical analysis and Monte-Carlo simulations have confirmed
 that the RZF-TCP aided and distance-based ACM scheme substantially outperforms our
 previous EB-TCP aided and distance-based ACM scheme. In the scenario where two
 communicating aircraft fly at a typical cruising speed of 920\,km/h in opposite
 direction all the way to the maximum horizon communication distance 740\,km, the
 RZF-TPC aided and distance-based ACM scheme is capable of transmitting about
 11.7\,GB and 16.5\,GB extra data volumes compared to EB-TCP aided and distance-based
 ACM scheme for the configurations of 32 DTAs/4 DRAs and 64 DTAs/4 DRAs, respectively.
 This study has therefore offered a practical high-rate, high-SE solution for air-to-air
 communications.

\appendix

\subsection{Gallery of Lemmas}\label{Apa}

\begin{lemma}[Matrix inversion lemma I \cite{silverstein1995empirical}]\label{L1}
 Given the Hermitian matrix $\bm{A}\in \mathbb{C}^{N\times N}$, vector $\bm{x}\in \mathbb{C}^N$
 and scalar $\tau\in \mathbb{C}$, if $\bm{A}$ and $\left(\bm{A}+\tau \bm{x}\bm{x}^{\rm H}\right)$
 are invertible, the following identity holds
\begin{align}\label{eA1}
 \left(\bm{A} + \tau \bm{x}\bm{x}^{\rm H}\right)^{-1} \bm{x} =&
  \frac{\bm{A}^{-1}\bm{x}}{1 + \tau \bm{x}^{\rm H}\bm{A}^{-1}\bm{x}} .
\end{align}
\end{lemma}

\begin{lemma}[Matrix inversion lemma II \cite{silverstein1995empirical}]\label{L2}
 Given the Hermitian matrix $\bm{A}\in \mathbb{C}^{N\times N}$, vector $\bm{x}\in \mathbb{C}^N$
 and scalar $\tau\in \mathbb{C}$, if $\bm{A}$ and $\left(\bm{A}+\tau \bm{x}\bm{x}^{\rm H}\right)$
 are invertible, the following identity holds
\begin{align}\label{eA2}
 \left(\bm{A} + \tau \bm{x}\bm{x}^{\rm H}\right)^{-1} =& \bm{A}^{-1} +
  \frac{\bm{A}^{-1}\tau \bm{x}\bm{x}^{\rm H}\bm{A}^{-1}}{1 + \tau \bm{x}^{\rm H}\bm{A}^{-1}\bm{x}} .
\end{align}
\end{lemma}

\begin{lemma}[\cite{zhang2017adaptive}]\label{L3}
 Let $\bm{A}\in \mathbb{C}^{N \times N}$ and $\bm{x}\sim \mathcal{CN}\left(\frac{1}{\sqrt{N}}\bm{m},
 \frac{1}{N}\bm{\Upsilon}\right)$, where $\frac{1}{\sqrt{N}}\bm{m}\in \mathbb{C}^N$ and
 $\frac{1}{N}\bm{\Upsilon}\in \mathbb{C}^{N \times N}$ are the mean vector and the covariance matrix
 of the random vector $\bm{x}\in \mathbb{C}^N$, respectively. Assuming that $\bm{A}$ has a uniformly
 bounded spectral norm with respect to $N$ and $\bm{x}$ is independent of $\bm{A}$,   we have
\begin{align}\label{eA3}
 \lim_{N\to \infty}\bm{x}^{\rm H}\bm{A}\bm{x} =& \text{Tr}\left\{\left(\frac{1}{N}\bm{M} +
 \frac{1}{N}\bm{\Upsilon}\right)\bm{A}\right\} ,
\end{align}
 where $\bm{M}=\bm{m}\bm{m}^{\rm H}$.
\end{lemma}

\begin{lemma}\label{L4}
 Let $\bm{A}\in \mathbb{C}^{N \times N}$, and two independent random vectors $\bm{x}\in \mathbb{C}^N$
 and $\bm{y}\in \mathbb{C}^N$ have the distributions $\bm{x}\sim \mathcal{CN}\left(\frac{1}{\sqrt{N}}\bm{m}_x,
 \frac{1}{N}\bm{\Upsilon}_x\right)$ and $\bm{y}\sim \mathcal{CN}\left(\frac{1}{\sqrt{N}}\bm{m}_y,
 \frac{1}{N}\bm{\Upsilon}_y\right)$, where $\frac{1}{\sqrt{N}}\bm{m}_x\in \mathbb{C}^N$ and
 $\frac{1}{\sqrt{N}}\bm{m}_y\in \mathbb{C}^N$ are the mean vectors, while $\frac{1}{N}\bm{\Upsilon}_x
 \in \mathbb{C}^{N \times N}$ and $\frac{1}{N}\bm{\Upsilon}_y\in \mathbb{C}^{N \times N}$ are 
 the covariance matrices of $\bm{x}$ and $\bm{y}$, respectively. Assuming that $\bm{A}$ has a
 uniformly bounded spectral norm with respect to $N$, and $\bm{x}$ and $\bm{y}$ are independent
 of $\bm{A}$, we have
\begin{align}\label{eA4}
 \lim_{N\to \infty}\bm{x}^{\rm H}\bm{A}\bm{y} =& \text{Tr}\left\{\frac{1}{N}\bm{M}_{xy}\bm{A}\right\} ,
\end{align}
 where $\bm{M}_{xy}=\bm{m}_{x}\bm{m}_{y}^{\rm H}$.
\end{lemma}

\begin{proof}
 Let $\bm{z}_{x}=\sqrt{N}\bm{x}-\bm{m}$. Since $\bm{x}\sim \mathcal{CN}\left(\frac{1}{\sqrt{N}}\bm{m}_x,
 \frac{1}{N}\bm{\Upsilon}_x\right)$, $\bm{z}_x\sim \mathcal{CN}\left(\bm{0}_N,\bm{\Upsilon}_x\right)$.
 Let $\bm{z}_y=\sqrt{N}\bm{y}-\bm{m}$. As $\bm{y}\sim \mathcal{CN}\left(\frac{1}{\sqrt{N}}\bm{m}_y,
 \frac{1}{N}\bm{\Upsilon}_y\right)$, $\bm{z}_x\sim \mathcal{CN}\left(\bm{0}_N,\bm{\Upsilon}_y\right)$.
 Furthermore,
\begin{align}\label{eA5}
 \bm{x}^{\rm H}\bm{A}\bm{y} =& \left(\frac{1}{\sqrt{N}}\bm{m}_x + \frac{1}{\sqrt{N}}\bm{z}_x\right)^{\rm H}
  \bm{A}\left(\frac{1}{\sqrt{N}}\bm{m}_y + \frac{1}{\sqrt{N}}\bm{z}_y\right) \nonumber \\
 =& \frac{1}{N}\bm{m}_x^{\rm H}\bm{A}\bm{m}_y + \frac{1}{N}\bm{z}_x^{\rm H}\bm{A}\bm{z}_y
  + \frac{1}{N}\bm{m}_x^{\rm H}\bm{A}\bm{z}_y + \frac{1}{N}\bm{z}_x^{\rm H}\bm{A}\bm{m}_y .
\end{align}
 Since $\bm{z}_x\sim \mathcal{CN}\left(\bm{0}_N,\bm{\Upsilon}_x\right)$ and $\bm{z}_y\sim
 \mathcal{CN}\left(\bm{0}_N,\bm{\Upsilon}_y\right)$, $\bm{z}_x$ and $\bm{z}_y$ do not
 depend on $\bm{m}_x$ and $\bm{m}_y$. According to Lemma~1 of \cite{fernandes2013inter}, we have
\begin{align}
 \lim\limits_{N\to \infty} \frac{\bm{m}_x^{\rm H}\bm{A}\bm{z}_y}{N} =& 0 , \label{eA6} \\
 \lim\limits_{N\to \infty} \frac{\bm{z}_x^{\rm H}\bm{A}\bm{m}_y}{N} =& 0 . \label{eA7}
\end{align}
 Since $\bm{z}_x$ and $\bm{z}_y$ are independent, according to the trace
 lemma of \cite{hoydis2012random}, we have
\begin{align}\label{eA8}
 \lim\limits_{N\to \infty} \frac{1}{N}\bm{z}_x^{\rm H}\bm{A}\bm{z}_y = 0 .
\end{align}
 Furthermore,
\begin{align}\label{eA9}
 \frac{1}{N}\bm{m}_x^{\rm H}\bm{A}\bm{m}_y = \text{Tr}\left\{\frac{1}{N}\bm{A}\bm{M}_{xy}\right\} .
\end{align}
 Taking the limit $N\to \infty$ as well as substituting (\ref{eA6}) to (\ref{eA9}) into (\ref{eA5})
 results in (\ref{eA4}).
\end{proof}

\subsection{Derivation of the Optimal Regularization Parameter}\label{Apb}

 Because the term $\mathcal{E}\bigg\{\sum\limits_{a=1}^A \Big\|\frac{1}{N_t}
 \sqrt{\frac{P_{r,b^*}^a}{P_{r,b^*}^{a^*}}}\bm{H}_{b^*}^a\bm{V}_{b^a}^a \bm{X}^a\Big\|^2
 + \Big\|\frac{1}{N_t\sqrt{P_{r,b^*}^{a^*}}}\bm{W}_{b^*}\Big\|^2\bigg\}$ is
 independent of $\xi_{b^*}^{a^*}$, 
\begin{align}\label{eB1}
 \frac{d\, \mathcal{J}\big(\xi_{b^*}^{a^*}\big)}{d\, \xi_{b^*}^{a^*}} =&
  \mathcal{E}\Bigg\{\bigg(\frac{1}{N_t}\bm{H}_{b^*}^{a^*}\bm{V}_{b^*}^{a^*} \bm{X}^{a^*} -
  \bm{X}^{a^*}\bigg)^{\rm H}\frac{d\, \big(\frac{1}{N_t}\bm{H}_{b^*}^{a^*}\bm{V}_{b^*}^{a^*}
  \bm{X}^{a^*} - \bm{X}^{a^*}\big)}{d\, \xi_{b^*}^{a^*}} \nonumber \\	
 &\hspace*{-15mm} + \frac{d\, \big(\frac{1}{N_t}\bm{H}_{b^*}^{a^*}\bm{V}_{b^*}^{a^*} \bm{X}^{a^*}  -
  \bm{X}^{a^*}\big)^{\rm H}}{d\, \xi_{b^*}^{a^*}} \Big(\frac{1}{N_t}\bm{H}_{b^*}^{a^*}
  \bm{V}_{b^*}^{a^*} \bm{X}^{a^*}  - \bm{X}^{a^*}\Big)\Bigg\} \nonumber \\
&\hspace*{-15mm} = \mathcal{E}\bigg\{\bigg(\frac{1}{N_t}\bm{H}_{b^*}^{a^*}\bm{V}_{b^*}^{a^*} \bm{X}^{a^*}
  - \bm{X}^{a^*}\bigg)^{\rm H}\bigg(- \frac{1}{N_t}\bm{H}_{b^*}^{a^*}\Big(\big(\bm{\Upsilon}_{b^*}^{a^*}\big)^2
  \big(\widehat{\bm{H}}_{b^*}^{a^*}\big)^{\rm H}\Big) \bm{X}^{a^*}\bigg) \nonumber \\
 &\hspace*{-15mm} + \bigg( -\frac{1}{N_t}\bm{H}_{b^*}^{a^*}\Big(\big(\bm{\Upsilon}_{b^*}^{a^*}\big)^2
  \big(\widehat{\bm{H}}_{b^*}^{a^*}\big)^{\rm H}\Big) \bm{X}^{a^*}\bigg)^{\rm H}\bigg(\frac{1}{N_t}
  \bm{H}_{b^*}^{a^*}\bm{V}_{b^*}^{a^*} \bm{X}^{a^*} - \bm{X}^{a^*}\bigg)\bigg\} \nonumber \\
&\hspace*{-15mm} = \mathcal{E}\bigg\{\!\! -\! \frac{1}{N_t^2}\text{Tr}\Big\{\! E_{\rm s}\widehat{\bm{H}}_{b^*}^{a^*}
  \bm{\Upsilon}_{b^*}^{a^*} \big(\bm{H}_{b^*}^{a^*}\big)^{\rm H} \bm{H}_{b^*}^{a^*}
  \big(\bm{\Upsilon}_{b^*}^{a^*}\big)^2 \big(\widehat{\bm{H}}_{b^*}^{a^*}\big)^{\rm H} \! \Big\}
  \! +\! \frac{1}{N_t}\text{Tr}\Big\{\! E_{\rm s}\bm{H}_{b^*}^{a^*}\big(\bm{\Upsilon}_{b^*}^{a^*}\big)^2\!
  \big(\widehat{\bm{H}}_{b^*}^{a^*}\big)^{\rm H}\! \Big\} \nonumber \\
 &\hspace*{-15mm} \!\! - \! \frac{1}{N_t^2}\text{Tr}\Big\{\! E_{\rm s}\widehat{\bm{H}}_{b^*}^{a^*}
  \big(\bm{\Upsilon}_{b^*}^{a^*}\big)^2
  \big(\bm{H}_{b^*}^{a^*}\big)^{\rm H} \bm{H}_{b^*}^{a^*}\bm{\Upsilon}_{b^*}^{a^*}
  \big(\widehat{\bm{H}}_{b^*}^{a^*}\big)^{\rm H} \! \Big\}	
  \! +\! \frac{1}{N_t}\text{Tr}\Big\{\! E_{\rm s}\widehat{\bm{H}}_{b^*}^{a^*}\big(\bm{\Upsilon}_{b^*}^{a^*}\big)^2
  \big(\bm{H}_{b^*}^{a^*}\big)^{\rm H}\! \Big\} \!\! \bigg\} \! ,\!
\end{align}
 where $E_{\rm s}=\mathcal{E}\big\{\big|X_{n_r}^{a^*}\big|^2\big\}$, $1\le n_r\le N_r$.
 Setting $\frac{d\, \mathcal{J}\big(\xi_{b^*}^{a^*}\big)}{d\, \xi_{b^*}^{a^*}}=0$ and
 followed by some further operations yields
\begin{align}\label{eB2}
 & \text{Tr}\bigg\{\bigg(\frac{1}{N_t}\widetilde{\bm{\Xi}}_{b^*}^{a^*} - \xi_{b^*}^{a^*}\bm{I}_{N_t}
  \bigg)\bm{\Upsilon}_{b^*}^{a^*}\big(\widehat{\bm{H}}_{b^*}^{a^*}\big)^{\rm H}\widehat{\bm{H}}_{b^*}^{a^*}
 \big(\bm{\Upsilon}_{b^*}^{a^*}\big)^2 \bigg\} = 0,
\end{align}
 where $\widetilde{\bm{\Xi}}_{b^*}^{a^*}=\text{diag}\big\{\widetilde{\varphi}_1,\cdots ,
 \widetilde{\varphi}_{N_t}\big\}$. This proves that (\ref{eq70}) is an optimal regularization parameter.

\end{document}